\documentclass[11pt]{amsart}

\usepackage{times}
\usepackage{natbib, subcaption, float}
\usepackage{amsmath, amssymb, amsfonts, amsthm}
\usepackage{tikz}
\usetikzlibrary{shapes,arrows,positioning}
\usepackage[colorlinks=TRUE,citecolor=blue,urlcolor=blue, linkcolor=blue]{hyperref}
\usepackage{graphicx, enumitem, bm,graphicx,psfrag,epsf,placeins,subcaption, setspace, comment}
\usepackage{tikz}
\usetikzlibrary{shapes,arrows,positioning}
\usepackage{algorithm, algpseudocode, natbib}
\usepackage[foot]{amsaddr}
\usepackage{thmtools}
\usepackage{thm-restate}

\usepackage[margin=3 cm]{geometry}

\allowdisplaybreaks

\newcommand{\E}{{\mathbb E}}       % for parents
\newcommand{\pa}{{\rm pa}}       % for parents
\newcommand{\an}{{\rm an}}	   % for ancestors
\newcommand{\An}{{\rm An}}	   % for ancestors
\newcommand{\de}{{\rm de}}	   % for desendents 
\newcommand{\var}{{\rm var}}       % for variance
\newcommand{\lMin}{\lambda_{\rm min}}       
\newcommand{\vertiii}[1]{{\left\vert\kern-0.25ex\left\vert\kern-0.25ex\left\vert #1 
		\right\vert\kern-0.25ex\right\vert\kern-0.25ex\right\vert}}
\newcommand\numberthis{\addtocounter{equation}{1}\tag{\theequation}}

\newtheorem{theorem}{Theorem}
\newtheorem{lemma}{Lemma}

\newtheorem*{remark}{Remark}
\newtheorem{condition}{Condition}

\newtheorem{corollary}{Corollary}

\setlist[itemize]{noitemsep}
\setlist[enumerate]{noitemsep}

\title{High Dimensional Causal Discovery Under Non-Gaussianity}
\author[Y.S. Wang]{Y. Samuel Wang$^1$}
\address{$^1$Booth School of Business\\
	University of Chicago\\
	Chicago, IL, U.S.A.\\}

\author[M. Drton]{Mathias Drton$^2$}
\email{ysamuelwang@gmail.com, md5@uw.edu}

\address{$^2$Department of Statistics\\
	University of Washington\\
	Seattle, WA, U.S.A.\\}

\address{$^2$Department of Mathematical Sciences\\
	University of Copenhagen\\
	Copenhagen, Denmark\\}
\begin{document}
	\maketitle
	
	\begin{abstract}
		  We consider graphical models based on a recursive system of linear structural equations.
          This implies that there is an ordering, $\sigma$, of the
          variables such that each observed variable $Y_v$ is a linear
          function of a variable specific error term and the other
          observed variables $Y_u$ with $\sigma(u) < \sigma (v)$.  The
          causal relationships, i.e., which other variables the linear
          functions depend on, can be described using a directed
          graph.  It has been previously shown that when the variable
          specific error terms are non-Gaussian, the exact causal
          graph, as opposed to a Markov equivalence class, can be
          consistently estimated from observational data.  We propose
          an algorithm that yields consistent estimates of the graph
          also in high-dimensional settings in which the number of
          variables may grow at a faster rate than the number of
          observations, but in which the underlying causal structure
          features suitable sparsity; specifically, the maximum
          in-degree of the graph is controlled.  Our theoretical
          analysis is couched in the setting of log-concave error
          distributions.
	\end{abstract}

\section{Introduction}

Prior work shows the possibility of causal discovery with
observational data in the framework of linear structural equation
models with non-Gaussian errors.  However, existing methods for
estimation of the causal structure are applicable only in
low-dimensional settings, in which the number of variables, $p$, is
small compared to the sample size, $n$.  In this paper, we develop a
method which, given suitable sparsity, recovers the exact causal
structure consistently in high-dimensional regimes where $p$ grows
along with $n$.  Careful considerations of computational aspects make
our method a practical and statistically sound exploratory tool for
the intended high-dimensional settings.
% We anticipate that this method will useful for practitioners
% interested in a statistically sound exploratory tool for discovering
% causal structure from high-dimensional data.

Let $Y_1,\dots,Y_n \in \mathbb{R}^p$ be multivariate data from an
observational study, specifically, the observations form an
independent, identically distributed sample.  We encode the causal
structure generating dependences in the underlying $p$-variate joint
distribution by a graph $G = (V, E)$ with vertex set
$V=\{1,\dots,p\}$. Each node, $v \in V$, corresponds to an observed
variable in $Y_i=(Y_{vi})_{v\in V}$, and each directed edge,
$(u, v) \in E$, indicates that $Y_{ui}$ has a direct causal effect on
$Y_{vi}$. Thus, positing causal structure is equivalent to selecting a
graph.  We will only consider directed acyclic graphs (DAGs),
directed graphs which do not contain directed cycles.  Given the
correspondence between a node $v \in V$ and the random variable
$Y_{vi}$, we will at times let $v$ stand in for $Y_{vi}$; for
instance, when stating stochastic independence relations.

Discovery of causal structure from observational data is difficult
because of the super-exponential set of possible models, some of which
may be indistinguishable from others. Despite this difficulty, many
methods for causal discovery have been developed and have seen
fruitful applications; see the recent review of
\citet{drton2017stucture}.  In particular, the celebrated PC algorithm
\citep{spirtes2000causation} is a constraint-based method which first
infers a set of conditional independence relationships and then
identifies the associated Markov equivalence class; this class
contains all DAGs compatible with the inferred conditional
independences. \citet{kalisch2007estimating} show if the maximum total
degree of the graph is controlled and the data is Gaussian, then the
PC algorithm can consistently recover the true Markov equivalence
class even in high-dimensional settings where the number of variables
grows with the number of samples. \citet{harris2013nonparanormal}
extend the result to Gaussian copula models using rank correlations.

However, graphs within the same Markov equivalence class may have drastically different causal and scientific interpretations.
For the graphs in Figure~\ref{fig:equivalenceClass}, conditional
independence tests can distinguish model (a) from the rest but
cannot distinguish models (b), (c), and (d) from each other.
% \citet{steinsky:2013} and \citet{he2015counting} give
Although \citet{maathuis2009ida} provide a procedure for bounding the
size of a causal effect over graphs within an equivalence class,
interpretation of the set of possibly conflicting graphs
can remain difficult.
Results on the size and number of Markov equivalence classes, which
may be exponentially large, can be found e.g.\ in \citet{steinsky:2013}.

% Within a Markov
% equivalence class, the presence or absence of an edge between any pair
% of nodes is the same for all graphs; however, in general, the
% orientation of the edges may differ, resulting in a set of possible
% graphs which grows exponentially with the number of nodes; see
% \cite{steinsky:2013} and \cite{he2015counting} for results on the size
% and number of Markov equivalence classes.
% Figure~\ref{fig:equivalenceClass} shows all 3 node graphs with two
% edges such that nodes 1 and 2 as well as nodes 2 and 3 are adjacent.
% Model (a) can be distinguished from models (b), (c), and (d) using
% conditional independence tests; however, models (b), (c), and (d) are
% mutually indistinguishable by conditional independence tests.

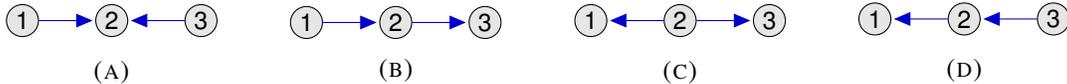
\begin{figure}[t]
	\centering
	\begin{subfigure}{0.23\textwidth}\centering
		\begin{tikzpicture}[->,>=triangle 45,shorten >=1pt,
		auto,
		main node/.style={ellipse,inner sep=0pt,fill=gray!20,draw,font=\sffamily,
			minimum width = .5cm, minimum height = .5cm, scale=.85}]
		
		\node[main node] (1) {1};
		\node[main node] (2) [right = .75cm of 1]  {2};
		\node[main node] (3) [right = .75cm of 2]  {3};
		
		\path[color=black!20!blue,style={->}]
		(1) edge node {} (2)
		(3) edge node {} (2);
		\end{tikzpicture}
		\caption{}
	\end{subfigure}
	\ 
	\begin{subfigure}{0.23\textwidth}\centering
		\begin{tikzpicture}[->,>=triangle 45,shorten >=1pt,
		auto,
		main node/.style={ellipse,inner sep=0pt,fill=gray!20,draw,font=\sffamily,
			minimum width = .5cm, minimum height = .5cm, scale=.85}]
		
		\node[main node] (1) {1};
		\node[main node] (2) [right = .75cm of 1]  {2};
		\node[main node] (3) [right = .75cm of 2]  {3};
		
		\path[color=black!20!blue,style={->}]
		(1) edge node {} (2)
		(2) edge node {} (3);
		\end{tikzpicture}
		\caption{}
	\end{subfigure}
	\  
	\begin{subfigure}{0.23\textwidth}\centering
		\begin{tikzpicture}[->,>=triangle 45,shorten >=1pt,
		auto,
		main node/.style={ellipse,inner sep=0pt,fill=gray!20,draw,font=\sffamily,
			minimum width = .5cm, minimum height = .5cm, scale=.85}]
		
		\node[main node] (1) {1};
		\node[main node] (2) [right = .75cm of 1]  {2};
		\node[main node] (3) [right = .75cm of 2]  {3};
		
		\path[color=black!20!blue,style={->}]
		(2) edge node {} (1)
		(2) edge node {} (3);
		\end{tikzpicture}
		\caption{}
	\end{subfigure}
	\  
	\begin{subfigure}{0.23\textwidth}\centering
		\begin{tikzpicture}[->,>=triangle 45,shorten >=1pt,
		auto,
		main node/.style={ellipse,inner sep=0pt,fill=gray!20,draw,font=\sffamily,
			minimum width = .5cm, minimum height = .5cm, scale=.85}]
		
		\node[main node] (1) {1};
		\node[main node] (2) [right = .75cm of 1]  {2};
		\node[main node] (3) [right = .75cm of 2]  {3};
		
		\path[color=black!20!blue,style={->}]
		(3) edge node {} (2)
		(2) edge node {} (1);
		\end{tikzpicture}
		\caption{}
	\end{subfigure}
	
	\caption{\label{fig:equivalenceClass}The Markov equivalence class of graph (a) is a singleton. However, graph (b), (c) and (d) are Markov equivalent and imply the same set of conditional independences.}
\end{figure}

In contrast, it has been shown that under various additional
assumptions, the exact graph structure, not just an equivalence class,
can be identified from observational data
\citep{loh2014inverse,peters2014equal,ernest2016causal}. In
particular, \citet{shimizu2006lingam} show this to be the case
under three main assumptions: (1) the data are generated by a linear
structural equation model, (2) the error terms in the structural
equations are non-Gaussian, and (3) there is no unobserved
confounding among the observed variables; i.e.,
errors are  independent. These assumptions yield the linear non-Gaussian acyclic model,
abbreviated as LiNGAM, which is described formally in
Section~\ref{sec:lingam}.  Under the LiNGAM framework, the four models
% (a), (b), (c), and (d)
from Figure~\ref{fig:equivalenceClass} are mutually distinguishable.
\citet{shimizu2006lingam} use independent component analysis to
estimate the graph structure, and the subsequent DirectLINGAM
\citep{shimizu2011direct} and Pairwise LiNGAM
\citep{hyvarinen2013pairwise} methods iteratively select a causal
ordering by computing pairwise statistics.  These methods are motivated by identifiability results that are derived by iteratively forming
conditional expectations.  In practice, the conditional expectations are
estimated using larger and larger regression models.  As a result, the
methods become inapplicable when the number of variables exceeds the
sample size.
% However, the statistical guarantees require infinite samples--i.e., population values--and the proposed methods  
% are inconsistent in scenarioshigh-dimensional settings where the number of
% variables scale as fast or faster than the sample size.
% In practice, the existing methods use empirical quantities and will not even run if there are more variables than observed samples.

We develop a modification of the DirectLiNGAM algorithm that is
suitable for high-dimensional data and 
% In order to remedy this problem, we extend the previous work by
% proposing a modified DirectLiNGAM algorithm and
give guarantees for when our algorithm will consistently recover
the true graph in high-dimensional asymptotic scenarios.
% , and is computationally tractable for large graphs.
Most notably, our analysis considers restricted maximum in-degree of
the graph and assumes log-concave distributions. The theory also
applies to hub graphs where the maximum out-degree may grow with the
size of the graph, which is in contrast to the conditions needed for
high-dimensional consistency of the PC algorithm
\citep{kalisch2007estimating}.  Hub graphs appear in many 
biological networks \citep{hao2012revisiting}.
% which do not satisfy the
% conditions
%Along the way, we also propose a new test statistic which encodes causal direction, is easy to analyze, and is computationally inexpensive. 

\section{Causal discovery algorithm}
\subsection{Generative model and notation}\label{sec:lingam}

%YSW Consider a $p\times n$ matrix of real-valued observations $Y=(Y_{vi})$ whose columns $Y_1, \ldots, Y_n$ are
% Consider real-valued $p$-variate observations $Y_1, \ldots, Y_n$ which are
%independent, identically distributed replications.  
We assume that the
observations $Y_1, \ldots, Y_n \in \mathbb{R}^p$ are independent, identically distributed replications generated from a linear structural equation model so
that the elements of each random vector $Y_i$ satisfy
\begin{equation}
\label{eq:sem}
Y_{vi} = \sum_{u \neq v} \beta_{vu} Y_{ui} + \varepsilon_{vi}, 
% \quad v=1,\dots,p,\; i=1,\dots,n,
\end{equation}
where the $\beta_{vu}$ are unknown real parameters that quantify the
direct linear effect of variable $u$ on variable $v$, and
$\varepsilon_{vi}$ is an error term %specific to $v$
of unknown
distribution $P_v$.  We assume $\varepsilon_{vi}$ has mean 0 and is
independent of all other error terms.  Our interest is in models that
postulate that a particular set of coefficients $\beta_{vu}$ is zero. In particular, the absence of an edge, $(u,v) \notin E$,
indicates that the model constrains the parameter $\beta_{vu}$ to
zero.  We assume that the graph, $G$, representing the model is a DAG,
which implies that the structural equation model is recursive; i.e.,
there exists a permutation of the variables, $\sigma$, such that
$\beta_{vu}$ is constrained to be zero unless $\sigma(u) <\sigma(v)$.

%Each one of the models we consider can be identified with a directed
%graph $G = (V,E)$ whose vertex set $V=\{1,\dots,p\}$ indexes the
%different variables.  The absence of an edge, $(u,v) \notin E$,
%indicates that the model constrains the parameter $\beta_{vu}$ to
% zero.  
% For recursive models, $G$ is an acyclic digraph.  
We denote the model given by graph $G$ by $\mathcal{P}(G)$.  Each
distribution $P\in \mathcal{P}(G)$ is induced through a choice of
linear coefficients $\left(\beta_{vu}\right)_{(u,v) \in E}$ and error
distributions $(P_v)_{v\in V}$.  Let $B=(\beta_{vu})$ be the
$p\times p$ matrix determined by the model constraints and the chosen
free coefficients.  Then the equations in~(\ref{eq:sem}) admit a
unique solution with $Y_i=(I-B)^{-1}\varepsilon_i$.  The error
vectors $\varepsilon_i=(\varepsilon_{vi})_{v\in V}$ are independent and identically distributed and
follow the product distribution $\otimes_{v\in V} P_v$.  The
distribution $P$ is then the joint distribution for $Y_i$ that is
induced by the transformation of $\varepsilon_i$.

Standard notation has the set $\pa(v) = \{u: (u,v) \in E\}$ comprise
the parents of a given node $v$.  The set of ancestors,
$\an(v)$, contains any node $u\not=v$ with a directed path from $u$
to $v$; we let $\An(v) = \an(v) \cup \{v\}$.   The set of
descendants, $\de(v)$, contains the nodes $u$ with
$v\in\an(u)$.

\subsection{Parental faithfulness}
An important approach to causal discovery begins by inferring
relations such as conditional independence and then determines graphs
compatible with empirically found relations.  For this approach to
succeed, %YSW it is important that 
the considered relations must correspond to structure in the graph $G$ as opposed to a special choice of parameters.
%YSW such as the coefficients $\beta_{vu}$.  
In the context of conditional independence, the assumption that any relation present in an
underlying joint distribution $P\in\mathcal{P}(G)$ corresponds to the
absence of certain paths in $G$ is known as the faithfulness
assumption; see \citet{uhler2013geometry} for a detailed discussion.
%YSW of this assumption. 
For %YSW the purpose of 
our work, we define a weaker
condition, %YSW which we refer to as 
parental faithfulness. In particular, if $u \in \pa(v)$, we require
that the total effect of $u$ on $v$ does not vanish when we modify the
considered distribution by regressing $v$ onto any set of its
non-descendants, as detailed next.

% For $v_j \in V$,
Let $l=(v_1, \ldots, v_z)$ be a directed path in $G$, so
$(v_j,v_{j+1})\in E$ for $j = 1, \ldots, z-1$.  Given coefficients
$\left(\beta_{vu}\right)_{(u,v) \in E}$, the path has weight
$ w(l) = \prod_{j = 1}^{z-1} \beta_{{v_{j+1}},v_{j}}$.  Let
$\mathcal{L}_{vu}$ be the set of all directed paths from $u$ to
$v$. Then the total effect of $u$ on $v$ is
$\pi_{vu} = \sum_{l \in\mathcal{L}_{vu}} w(l)$, with $\pi_{vu}=0$ if
$u\not\in\An(v)$ and $\pi_{vu} = 1$ if $u=v$.  The effect gives the
conditional mean of $v$ under interventions on $u$; i.e., $\pi_{vu} = \E(Y_{vi} \mid \text{\rm do}(Y_{ui} = y + 1)) -\E(Y_{vi}
\mid \text{\rm do}(Y_{ui} = y))$ using the
do-operator of \citet{pearl2009Causality}.
% we have
% $\pi_{vu} = \E(Y_{vi} \mid \text{\rm do}(Y_{ui} = y + 1)) -\E(Y_{vi}
% \mid \text{\rm do}(Y_{ui} = y))$.
% using the do-operator notation of \citet{pearl2009Causality}.  
% If $u\not\in\An(v)$, then $\pi_{vu}=0$, and by
% convention $\pi_{vv} = 1$.
Total effects may be calculated
% without enumerating all paths
by matrix inversion, % with
% inverting $I - B$; indeed,
$\Pi= \left(\pi_{vu}\right)_{u,v \in V} = (I- B)^{-1}$.

Let $\Sigma = \E(Y_i Y_i^t)$ be the covariance matrix of
the, for convenience, centered random vector $Y_i\sim P$.  Let
$\Sigma_{CC}$ be the principal sub-matrix for a non-empty set of indices
$C\subseteq V$.  For $v\in V\setminus C$, let $\Sigma_{Cv}$ be the
sub-vector comprised of the entries in places $(c,v)$ for $c\in C$.
Let
\begin{equation}
\label{eq:pop-reg-coeffs}
\beta_{vC} \;=\; (\beta_{vc.C})_{c\in C} =
(\Sigma_{CC})^{-1}\Sigma_{Cv}
\end{equation}
be the population regression coefficients when $v$ is regressed onto
$C$. The quantity $\beta_{vc.C}$ is defined even if $(c,v) \not\in E$,
and in general $\beta_{vc.C} \neq \beta_{vc}$ even if $(c, v) \in E$.
A pair $(u,v)\in E$ is parentally faithful if for any set
$C\subseteq V\setminus \left[\de(v)\cup\{ v,u \}\right]$, the residual
total effect defined as
\begin{equation}\label{eq:residualTotalEffect}
\pi_{vu.C} \;=\; \pi_{vu} - \sum_{c \in C}\beta_{vc.C}\pi_{cu}
\end{equation}
is nonzero.
If this holds for every pair $(u,v)\in E$, we say that the joint
distribution $P$ is parentally faithful with respect to $G$.
%YSW The condition trivially implies $\beta_{vu} \neq 0$ for $(u, v) \in E$.  
Parental faithfulness only pertains to the
linear coefficients and error variances, and the choices for which
parental faithfulness fails form a set of Lebesgue measure zero.
%Note that the coefficients and error variances can be identified from the covariance matrix of $P$.
The concept is exemplified in Figure~\ref{fig:parentalFaithfulness}.
\begin{figure}[t]
	\centering
	\begin{subfigure}[b]{0.4\textwidth}\centering
		\begin{tikzpicture}[->,>=triangle 45,shorten >=1pt,
		auto,
		main node/.style={ellipse,inner sep=0pt,fill=gray!20,draw,font=\sffamily,
			minimum width = .5cm, minimum height = .5cm, scale = .85}]
		
		\node[main node] (1)  {1};
		\node[main node] (2) [right = 1.2cm of 1]  {2};
		\node[main node] (3) [right = 1.2cm of 2] {3};
		
		\path[color=black!20!blue,style={->}]
		(1) edge[bend left = 40] node {} (3)
		(1) edge node {} (2)
		(2) edge node {} (3);
		\end{tikzpicture}
		\caption{}
	\end{subfigure}
	~
	\begin{subfigure}[b]{0.4\textwidth}
		\centering
		\begin{tikzpicture}[->,>=triangle 45,shorten >=1pt,
		auto,
		main node/.style={ellipse,inner sep=0pt,fill=gray!20,draw,font=\sffamily,
			minimum width = .5cm, minimum height = .5cm, scale = .85}]
		
		\node[main node] (1) {1};
		\node[main node] (2) [right = 1.2cm of 1]  {2};
		\node[main node] (3) [right = 1.2cm of 2]  {3};
		\node[main node] (4) [right = 1.2cm of 3]  {4};
		
		\path[color=black!20!blue,style={->}]
		(1) edge[bend left = 30] node {} (3)
		(1) edge[bend left = 30] node {} (4)
		(2) edge node {} (3)
		(2) edge[bend left = 30] node {} (4);
		\end{tikzpicture}
		\caption{}
	\end{subfigure}
	\caption[Example of parental faithfulness]{\label{fig:parentalFaithfulness}In (a), the choice
		$\beta_{31} = \beta_{21} = 1$ and $\beta_{32} = -1$ results in
		parental unfaithfulness because $\pi_{31.\emptyset} = 0$. Also, the choice $\beta_{31} = \beta_{21} = \beta_{32} = 1$ and $\E(\varepsilon_1^2) = \E(\varepsilon_2^2) = \E(\varepsilon_3^2) = 1$ is not faithful because the partial correlation of $2$ and $1$ given $3$ is $0$, but is still parentally faithful. In (b),
		the choice $\beta_{31} = \beta_{32} = \beta_{42} = 1$, $\beta_{41} = 2$, and $\E(\varepsilon_1^2) = \E(\varepsilon_2^2) = \E(\varepsilon_3^2) = 1$ results in parental unfaithfulness because $\pi_{42.3} = 0$.}
\end{figure}
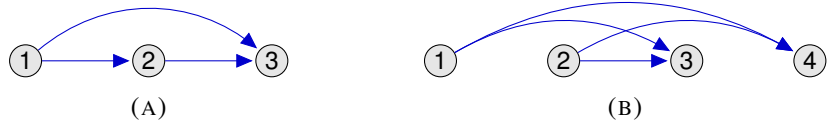

\subsection{Test statistic}\label{sec:testStat}
Reliable determination of the causal direction between $u$ and $v$
generally requires removal of all confounding. Thus,
\citet{shimizu2011direct} and \citet{hyvarinen2013pairwise} adjust $v$
and $u$ for all $x$ such that $\sigma(x) <\sigma(v)$ and
$\sigma(x) < \sigma(u)$. However, %YSW this results in
adjusting by an increasingly larger set of variables propagates error
proportional to the number of variables, rendering high-dimensional
estimation inconsistent, or impossible when the size of the adjustment
set exceeds the sample size. On the other hand, restricting the size
of the adjustment sets may not remove confounding completely.  The
method we present solves this problem via a statistic that is
conservative in the sense that it does not mistakenly certify causal
direction when confounding is present.

\citet{shimizu2011direct} calculate the kernel-based mutual
information between $v$ and the residuals of $u$ when it is regressed
onto $v$.  %YSW Under the suitable version of faithfulness,
The
corresponding population information is positive if and only if $v\in\de(u)$ or there is uncontrolled confounding between $v$
and $u$, that is, $u$ and $v$ have a common ancestor even when certain
edges are removed from the graph.
% the possible edge $(u,v)$.  
Hence, the mutual information can be used to test the hypothesis that
$v\not\in\de(u)$ versus the hypothesis that $v\in\de(u)$ or there is
confounding between $u$ and $v$. Unfortunately, calculating the mutual
information can be computationally burdensome, so
\citet{hyvarinen2013pairwise} propose a different parameter
$R_{vu}$. Without confounding, $R_{vu}>0$ if $v\in\an(u)$ and
$R_{vu}<0$ if
$u\in\an(v)$.  With confounding, however, the parameter can take
either sign, so it cannot be reliably used if we remain uncertain
about whether or not confounding occurs. We introduce a parameter that
shares the favorable properties of the mutual information % used by \citet{shimizu2011direct},
but admits computationally inexpensive estimators that are rational
functions of the sample moments of $Y$, which facilitates analysis of error propagation. 

The parameter we consider is motivated by the following
observation. Suppose the true generating mechanism is
$Y_1 \rightarrow Y_2$ so that $Y_1 = \varepsilon_1$ and
$Y_2 = \beta_{21}Y_1 + \varepsilon_2 $ for
$\varepsilon_1$ independent of $\varepsilon_2$. When the causal direction is
correctly specified, the linear coefficient $\beta_{21}$ is recovered
by $\mathbb{E}(Y_1^{K-1} Y_2) /\mathbb{E}(Y_1^{K})$ for all integers
$K$ greater than 1 for which $\mathbb{E}(Y_1^{K}) \neq 0$. Of course, letting $K = 2$ gives the typical least
squares estimator. This leads to the identity
$\mathbb{E}(Y_1^{K-1} Y_2)/\mathbb{E}(Y_1^{K}) =
\mathbb{E}(Y_1Y_2)/\mathbb{E}(Y_1^2)$ which implies
$\mathbb{E}(Y_1^{K-1}Y_2)\mathbb{E}(Y_1^2)- \mathbb{E}(Y_1 Y_2)
\mathbb{E}(Y_1^{K})= 0$, which holds even when $\mathbb{E}(Y_1^{K}) = 0$. In general, however, when the errors are
non-Gaussian and the roles of $Y_1$ and $Y_2$ are swapped, this
identity does not hold. When there are more than 2 variables involved,
we reduce the problem to a bivariate problem by conditioning on an
appropriate set $C$.  To this end, define for any $i$ the residual
\[
  Y_{vi.C} = Y_{vi} - \sum_{c\in C}\beta_{vc.C}Y_{ci},
\]
where
$\beta_{vc.C}$ are the population regression coefficients
from~\eqref{eq:pop-reg-coeffs}. When $C = \emptyset$, let
$Y_{vi.\emptyset} = Y_{vi}$.

\begin{theorem} \label{ch3:thm:testStat}
	Let $P \in \mathcal{P}(G)$ be a
	distribution in the model given by a DAG $G$, and let
	$Y_i \sim P$.  For $K > 2$, two distinct nodes $u$ and $v$, and
	any set $C\subseteq V\setminus\{u,v\}$, define
	\begin{equation} \label{eq:tauDefinition}
	\tau^{(K)}_{v.C\rightarrow u} \;=\;  \E_P(Y_{vi.C}^{K-1} Y_{ui})\E_P(Y_{vi.C}^2) - \E_P(Y_{vi.C}^K) \E_P(Y_{vi.C} Y_{ui}).
	\end{equation}
	\begin{enumerate}[label={\rm(\roman*)}]
		%[label={\bf(\roman*)} ,ref={\thetheorem(\roman*)}]
		\item \label{imp:tauZero}If $u \not\in \pa(v)$ and
		$\pa(v) \subseteq C\subseteq V\setminus \left[\de(v)\cup\{ v,u \}\right]$, then $\tau^{(K)}_{v.C\rightarrow u} = 0$. 
              \item \label{imp:tauNonZero} Suppose $u \in \pa(v)$ with
                $u,v$ parentally faithful under the covariance matrix
                of $P$.  If
                $C\subseteq V\setminus \left[\de(v)\cup\{ v,u
                  \}\right]$, then
                $\tau^{(K)}_{v.C\rightarrow u} \neq 0$ for generic
                error moments of order $3,\dots,K$.
	\end{enumerate}
\end{theorem}

Estimators $\hat \tau_{v.C \rightarrow u}^{(K)}$ of the parameter
from~\eqref{eq:tauDefinition} are naturally obtained from empirical
regression coefficients and empirical moments.

In Theorem \ref{ch3:thm:testStat}\ref{imp:tauNonZero}, the term
generic indicates that the set of error moments for which this
statement does not hold has Lebesgue measure zero. Given that there is a
finite number of sets $C\subset V$, the union of all exceptional sets
is also a null set.  A detailed proof of Theorem~\ref{ch3:thm:testStat} is
included in the supplement.  Claim \ref{imp:tauZero} can be shown via
direct calculation, and we give a brief sketch of \ref{imp:tauNonZero}
here. For fixed coefficients $(\beta_{vu})_{(u,v) \in E}$ and set
$C\subset V$, $\tau^{(K)}_{v.C\rightarrow u}$ is a rational function
of the error moments. Thus existence of a single choice of error moments
for which $\tau^{(K)}_{v.C\rightarrow u} \neq 0$ is sufficient to show
that the statement holds for generic error moments.  As the argument
boils down to showing that a certain polynomial is not the zero
polynomial \citep{okamoto1973distinctness}, the choice considered need
not necessarily be realizable by a particular distribution.  In
particular, we choose all moments of order less than
$K$ equal to those of the centered Gaussian distribution
with variance $\sigma_v^2 = \E(\varepsilon_v^2)$, but for the $K$th
moment we add an offset $\eta_v > 0$, so
\begin{equation}
\E(\varepsilon_v^K) =
\begin{cases}
\eta_v & \text{ if } K \text{ is odd},\\
(K-1)!!\sigma_v^K + \eta_v &\text{ if } K \text{ is even},
\end{cases}
\label{eq:nongauss-moments}
\end{equation}
where $q!! = \prod_{z = 0}^{\lceil q/2\rceil -1}(q - 2z) $ is the double factorial of $q$.  If there is no
confounding between $Y_{v.C}$ and $Y_{u}$, that is, no ancestor of $u$
is the source of a directed path to $v$ that avoids $C\cup\{u\}$, then
\begin{equation}\label{eq:proofSketch}
\tau^{(K)}_{v.C\rightarrow u} \;=\;
\pi_{vu.C}\left(\pi_{vu.C}^{K-2}\,\eta_u\sigma_{v}^2 -
\eta_v\sigma_{u}^2\right)
% \pi_{vu.C}^{K-1}\eta_u\sigma_{v}^2 - \pi_{vu.C}\eta_v\sigma_{u}^2\\
\end{equation}
with $\pi_{vu.C} \neq 0$, 
%YSW being the residual total effect from \eqref{eq:residualTotalEffect}.  
by the assumed parental faithfulness. %YSW ,$\pi_{vu.C}\not=0$
Thus, a choice of offsets with $\pi_{vu.C}^{K-2}\,\eta_u\sigma_{v}^2\not=
\eta_v\sigma_{u}^2$
% $\eta_{u} \neq \left(\eta_v\sigma^2_u\right)
% /\left(\pi_{vu.C}^{K-2}\sigma^2_v\right)$
implies $\tau^{(K)}_{v.C \rightarrow u} \neq 0$. A more involved
but similar argument can be made in the case of confounding. Under a slightly stronger form of faithfulness, $\tau^{(k)}_{v.C \rightarrow u} \neq 0$ if there is confounding regardless of whether $u \in \pa(v)$; see supplement Remark 1.

\begin{corollary}\label{ch3:thm:allGauss}
	Let $P_v$ and $P_u$ be two distributions that each have all moments up to order $K$ equal to those of some Gaussian distribution. Then there exists a graph $G$, for which $u \in \pa(v)$, and distributions $P$ which are parentally faithful with respect to $G$, but $\tau^{(K)}_{v.C\rightarrow u} = 0$
	for some set $C\subseteq V\setminus \left[\de(v)\cup\{ v, u\}\right]$.
	
\end{corollary}
\begin{proof}
	The moments of $P_v$ and $P_u$ satisfy~\eqref{eq:nongauss-moments}
	with $\eta_v = \eta_u = 0$.  Consequently, if there exists a set $C$ such that there is no confounding between $Y_{v.C}$ and $Y_u$, then
	$\tau_{v.C \rightarrow u}$ satisfies \eqref{eq:proofSketch}, the
	right-hand side of which is zero when $\eta_v = \eta_u = 0$. For example, if $\varepsilon_v$ and $\varepsilon_u$ are both Gaussian and the graph is $u \rightarrow v$, $\tau^{(K)}_{v \rightarrow u} = 0$ for all choices of $\beta_{vu}$ and all $K$.
\end{proof}
Corollary~\ref{ch3:thm:allGauss} confirms that the null set to be avoided in Theorem~\ref{ch3:thm:testStat}\ref{imp:tauNonZero} contains
points for which all error moments are consistent with some Gaussian distribution. Thus, our identification
of causal direction requires that the error moments of order at most $K$ be inconsistent with all
Gaussian distributions. In practice, we consider the case $K = 3, 4$ and recommend $K=4$ unless one is certain the errors are not symmetric. We refer readers to \citet{hoyer2008arbitrary} for a full characterization of when graphs with both Gaussian and non-Gaussian
errors are identifiable.

At each step, the high-dimensional LiNGAM algorithm presented in
Section~\ref{sec:alg} considers a sub-graph and searches for a root
node, i.e., a node without any parents. Suppose $\de(V_2) = V_2 \subseteq V$; if $v \in V_2$ is not a root in the sub-graph induced by
$V_2$, then there must exist some
$u \in V_2$ with $\tau_{v.C \rightarrow u} \neq 0$ for all sets $C$
which are upstream of $v$ and $V_2$. However, if $v$ is a root, then
$\tau_{v.C \rightarrow u} = 0$ for all $u \in V_2$ when $C =
\pa(v)$. Thus, to test whether $v$ is a root in $V_2$, we aggregate
the various $\tau$ parameters corresponding to $u \in V_2$ and
conditioning sets $C$. Corollary~\ref{ch3:thm:aggStat} describes two
ways to do this aggregation. If $v$ is a root, this quantity will be $0$, and if $v$ is not a root, it will be positive. 
 
\begin{corollary}\label{ch3:thm:aggStat}
	Let $P \in \mathcal{P}(G)$, let $v\in V$, and consider two disjoint
	sets $V_1, V_2\subseteq V\setminus\{v\}$.  For a chosen non-negative
	integer $J$, define
	\begin{align*}\label{eq:min_max}
	T^{(K)}_1(v, V_1, V_2) &\;=\; \min_{C \in V_1(J)}  \max_{u \in V_2} |\tau^{(K)}_{v.C \rightarrow u}|, \;\qquad
	T^{(K)}_2(v, V_1, V_2) &\;=\; \max_{u \in V_2} \min_{C \in V_1(J)}  |\tau^{(K)}_{v.C \rightarrow u}|,
	\end{align*}
	where $V_1(J)=\{C\subseteq V_1: |C| = J\}$ if $J\le |V_1|$ and
	$V_1(J)=V_1$ if $J\ge |V_1|$.

	\begin{enumerate}[label={\rm(\roman*)}] %,ref={\thecorollary(\roman*)}]
		\item \label{imp:aggStatZero}
		If $|\pa(v)| \leq J$ and $\pa(v) \subseteq V_1\subseteq V\setminus\de(v)$, then
		\[
		T^{(K)}_1(v, V_1, V_2) = T^{(K)}_2(v, V_1, V_2) =0.
		\]
		
		\item \label{imp:aggStatNonZero} Suppose $\beta_{vu} \neq 0$ for
		all $u\in\pa(v)$.  If $\de(V_2\cup \{v\})\subseteq V_2\cup \{v\}$ and
		$\pa(v) \cap V_2 \neq \emptyset$, then for generic error moments of
		order up to $K$, we have $ T^{(K)}_1(v, V_1, V_2)>0$ and
		$T^{(K)}_2(v, V_1, V_2) > 0$.
	\end{enumerate}
\end{corollary}
\begin{proof}
	\ref{imp:aggStatZero} The 
	statement  follows immediately from
	Theorem~\ref{ch3:thm:testStat}.  \ref{imp:aggStatNonZero} Since, $\pa(v) \cap V_2 \neq \emptyset$, but $\de(V_2\cup \{v\})\subseteq V_2\cup \{v\}$, there exists some $u \in \pa(v) \cap V_2$ such that $\de(u) \cap \pa(v) = \emptyset$. For that $u$ and 
	any $C \subseteq V_1$, the residual total effect is $
	\pi_{vu.C} = \beta_{vu} - \sum_{c \in C}\beta_{vc.C}\pi_{cu} =  \beta_{vu}$
	because the assumed facts $\de(V_2\cup \{v\})\cap V_1 = \emptyset$ and $\de(v)\cap\pa(v) = \emptyset$  imply that
	$\pi_{cu} = 0$ for all $c \in C$ and  $\pi_{vu} = \beta_{vu}$. We have assumed $\beta_{vu} \neq 0$, so, by Theorem~\ref{ch3:thm:testStat},
	generic error moments ensure that
	$|\tau^{(K)}_{v.C \rightarrow u}| > 0$ for all $C$, which in turn implies $T^{(K)}_j(v, V_1, V_2) > 0$ for $j = 1,2$.
\end{proof}
When (i) is satisfied, there may be more than one set $C$ which makes all pairwise statistics $0$. $T_1$ is calculated by finding a single conditioning set $C$ which minimizes the maximum pairwise statistic $\tau$ across all $u \in V_2$; in contrast, $T_2$ allows for a different conditioning set for each $u$. For fixed $v$, $V_1$, and $V_2$, the signs, either positive or zero, of $T_1$ (min-max) and $T_2$ (max-min) will always agree, but when $\pa(v) \cap V_2 \neq \emptyset$ and both quantities are positive, $T_1 \geq T_2$. Thus, the min-max statistic may be more robust to sampling error when testing if the parameters are non-zero. However, as discussed in Section~\ref{sec:alg}, $T_2$ can be computed more efficiently than $T_1$.

Theorem~\ref{ch3:thm:testStat}\ref{imp:tauNonZero} requires
parental faithfulness since we consider arbitrary $u \in \pa(v)$, whereas
Corollary~\ref{ch3:thm:aggStat}\ref{imp:aggStatNonZero} only requires that
$\beta_{vu} \neq 0$ since we maximize over $V_2$. %YSW a set which
% contains all of its own descendants.
Use of sample moments yields estimates $\hat \tau_{v.C \rightarrow
	u}$, which in turn yields estimates $\hat T^{(K)}_j(v, V_1, V_2)$ of $T^{(K)}_j(v, V_1, V_2)$ for
$j=1,2$.
% can be calculated from the sample moment based estimates
% $\hat \tau_{v.C \rightarrow u}$.
In the remainder of the paper, we drop the subscript $j$ %YSW $1$ or $2$
% and simply write $T$ or $\hat T$
in statements that apply to both
parameters/estimators.  Moreover, as we
always %YSW will always
fix $K$,
%YSW, in later sections we also lighten notation by omitting
we lighten notation by omitting the superscript, writing
$T(v, V_1, V_2)$, $\tau_{v.C \rightarrow u}$ and
$\hat \tau_{v.C \rightarrow u}$.

\section{Graph estimation procedure} \label{sec:graph_recovery}
\subsection{Algorithm}\label{sec:alg}
We now present a modified DirectLiNGAM algorithm which estimates the
underlying causal structure (Algorithm~\ref{alg:topOrder}). As in the
original algorithm, we identify a root and recur on the sub-graph that
has the identified root removed. After step $z$, we have a $z$-tuple,
$\Theta^{(z)}$, which gives an ordering of the roots identified so
far, and the remaining nodes $\Psi^{(z)} = V \setminus \Theta^{(z)}$.
In contrast to DirectLiNGAM, the proposed algorithm does not adjust
for all non-descendants, but only for subsets of limited size.  This
gives meaningful regression residuals also when the number of
variables exceeds the sample size and limits error propagation from
the estimated linear coefficients.

At each step $z$, we consider subsets of $\mathcal{C}^{(z)}_v \subseteq \Theta^{(z-1)}$, which we use to denote the set of possible parents for $v$. Naively allowing $\mathcal{C}^{(z)}_v = \Theta^{(z-1)}$ is not precluded by theory, but the number of subsets
$C \subset \Theta^{(z-1)}$ such that $|C| = J$ grows at
$\mathcal{O}(z^J)$. %YSW presenting an enormous computational effort even for moderate values of $p$ and $J$. 
Thus, for computational reasons, we  prune nodes which are not parents of $v$ by letting
%\begin{equation}\label{eq:parentSelection}
%\mathcal{C}^{(z)}_v = \Bigg\{p \in \Theta^{(z-1)}:\min_{C \subseteq \Theta^{(z-1)} \setminus \{p\}; |C| \leq J} |\hat \tau_{v.C \rightarrow p}| > g^{(z)}\Bigg\} \;\cup\; \Theta^{(z-1)}_{z-1},
%\end{equation}
\begin{equation}\label{eq:parentSelection}
\mathcal{C}^{(z)}_v = \Bigg\{p \in \mathcal{C}^{(z-1)}_v: \min_{C \in D^{(z)}_v} |\hat \tau_{v.C \rightarrow p}| > g^{(z)}\Bigg\} \;\cup\; \Theta^{(z-1)}_{z-1}
\end{equation}
where $D^{(z)}_v = \bigcup_{d < z} \{C: C \subseteq \mathcal{C}_v^{(d)} \setminus \{p\}; |C| \leq J \}$,  $\Theta^{(z-1)}_{z-1}$ is the node selected at the previous step, and $g^{(z)}$ is some cut-off value. Selecting a good value for $g^{(z)}$ is difficult because it should depend on the unknown signal strength. However, under the assumptions of Theorem~\ref{ch3:thm:deterministicCorrect}, if $r$ is the root selected at step $z-1$ and $\alpha$ is some tuning parameter in $[0,1]$, then letting $g^{(z)} = \max(g^{(z-1)}, \alpha \hat T(r, \mathcal{C}^{(z)}_r, \Psi^{(z-1)}) )$ will not mistakenly prune parents from $\mathcal{C}^{(z)}_v$. In Algorithm~\ref{alg:topOrder},
we do not update $\mathcal{C}^{(z)}_v$ after $v$ is selected as a root. Since the final cut-off, $g^{(p)}$, may be larger than the cut-off used to
select $\mathcal{C}^{(z)}_v$, a final pruning step uses the criteria
from~\eqref{eq:parentSelection} with $g^{(p)}$ to prune away nodes in $\mathcal{C}^{(p)}_v$ that may be ancestors but not parents.

A larger value of $\alpha$ prunes more aggressively, decreasing the computational
effort. However, setting $\alpha$ too large could result in incorrect
estimates if some parent of $v$ is incorrectly pruned from
$\mathcal{C}^{(z)}_v$. Section \ref{sec:deterministicSection}
discusses selecting an appropriate $\alpha$ and a more detailed discussion of computational savings from the pruning procedure is given in the supplement. 

%We say that $\Theta^{(z)}$ is consistent with a valid ordering of $G$ if for every $s,t \leq z$, $s < t$ only 
%if $\Theta^{(z)}_t \notin \an(\Theta^{(z)}_s)$ and $\Theta^{(z)} \cap \de(\Psi^{(z)}) = \emptyset$.  
\begin{algorithm}[htb]
	\caption{Estimate Causal DAG}\label{alg:topOrder}
	\begin{algorithmic}[1]
		\State Set $\Theta^{(0)} = \emptyset$ and $\Psi^{(0)} = [p]$
		\For{$z = 1, \ldots, p$}
		\For{$v \in \Psi^{(z-1)}$}
		\State Select the set of possible parents $\mathcal{C}^{(z)}_v \subseteq \Theta^{(z-1)}$ and compute $\hat T(v, \mathcal{C}^{(z)}_v, \Psi^{(z-1)}\setminus \{v\})$
		\EndFor
		\State Let $r = \arg\min_{v \in \Psi^{(z-1)}} \hat T(v, \mathcal{C}^{(z)}_v, \Psi^{(z-1)}\setminus \{v\})$
		\State Append $r$ to $\Theta^{(z-1)}$ to form $\Theta^{(z)}$ and set $\Psi^{(z)} = \Psi^{(z-1)} \setminus \{r\}$.
		\EndFor
		\State Prune ancestors to form parents $\mathcal{C}^{\star}_v$ for all $v \in V$
		\State \textbf{Return:} $\Theta^{(p)}$ as the topological ordering; $\{\mathcal{C}^{\star}_v\}_{v \in V}$ as the set of parents
	\end{algorithmic}
\end{algorithm}

As discussed in Section~\ref{sec:testStat}, $T_1$ may be more robust
to sampling error than $T_2$ but comes at greater
computational cost. At each step, $\Psi^{(z)}$ decreases by a single
node and $\mathcal{C}^{(z)}_v$ may grow by one node. If the
$|\Psi^{(z)}|^2$ values of
$\min_{C \in \mathcal{C}^{(z-1)}_v} \hat \tau_{v.C \rightarrow u}$
have been stored, updating $\hat T_2$, the max-min, only requires
testing the $\binom{|\mathcal{C}^{(z-1)}|}{J-1}$ subsets of
$\mathcal{C}^{(z)}_v$ which include
%YSW $\mathcal{C}^{(z-1)}_v \setminus \mathcal{C}^{(z)}_v$,
the variable selected at the previous step. Updating the min-max
statistic  $\hat T_1$ without redundant computation would require
storing the $\mathcal{O}\left((p-z)^2 z^J\right)$ values of $|\tau_{v.C
  \rightarrow u}|$. In practice, we completely recompute it at each
step. Section~\ref{highD:sec:simulations} demonstrates this trade-off
between computational burden and robustness.
% on simulated data. 

\subsection{Deterministic statement}\label{sec:deterministicSection}

Theorem~\ref{ch3:thm:deterministicCorrect} below makes a deterministic
statement about sufficient conditions under which
Algorithm~\ref{alg:topOrder} will output a specific graph $G$ when
given data $Y = (Y_1, \ldots, Y_n)$. We assume each $Y_i \sim P_Y$ but
allow model misspecification so that $P_Y$ may not be in $\mathcal{P}(G)$ for any DAG $G$. 
%YSW any acyclic digraph $G$. 
However, we require that the sample moments of $Y$ are close enough to the population moments for some distribution $P \in \mathcal{P}(G)$. For notational convenience, for $H \subseteq V$ and $\alpha \in \mathbb{R}^{|H|}$, let $\hat m_{H, \alpha} = \frac{1}{n}\sum_{i}^n \left(\prod_{v \in H}Y_{vi}^{\alpha_v}\right)$ denote a sample moment estimated from data $Y$, and let $m_{H, \alpha} = \E_P\left(\prod_{v \in H} Z_v^{\alpha_v}\right)$ denote a population moment for $Z \sim P$.

\begin{condition}\label{con:existence}
	For some $p$-variate distribution $P$, there exists a DAG $G$ with $|\pa(v)| \leq J$ for all $v \in V$ such that:
	\begin{enumerate}[label={(\alph*)},ref={\thetheorem(\alph*)}]
		\item \label{con:existenceNonZero}For all $v,u \in V$ and $C\subseteq V \setminus \{u,v\}$ with $|C| \leq J$ and $C\cap \de(v) = \emptyset$; if $u \in \pa(v)$ then the population quantities for $P$ satisfy $\left|\tau^{(K)}_{v.C\rightarrow u}\right| > \gamma > 0$.
		\item \label{con:existenceZero}For all $v, u \in V$ and $ C \subseteq V \setminus\{v,u\}$ with $|C| \leq J$ and $\pa(v) \subseteq C\subseteq V \setminus \de(v)$, if $u \not\in \pa(v)$, then the population quantities for $P$ satisfy $\tau^{(K)}_{v.C\rightarrow u}  = 0$.
	\end{enumerate}
\end{condition}
\begin{condition}\label{con:min_eigen}
	All $J \times J$ principal submatrices of the population covariance of $P$ have minimum eigenvalue greater or equal to $\lMin > 0$.
\end{condition}
\begin{condition}\label{con:bounded_moments}
	All population moments of $P$ up to degree $K$, $m_{V, \alpha}$ for $\sum_{v}\alpha_v \leq K$, are bounded by a constant $\infty > M > \max(1, \lMin / J)$ for positive integer $J$.
\end{condition}
\begin{condition}\label{con:sample_moments}
	All sample moments of $Y$ up to degree $K$, $\hat m_{V, \alpha}$ for $\sum_{v}\alpha_v \leq K$, are within $\delta_1  < \lMin /(2J)$ of the corresponding population values of $P$.
\end{condition}

The constraint in Condition~\ref{con:bounded_moments} that $M > \max(1, \lMin / J)$ is only used to facilitate simplification of the error bounds and is not otherwise necessary. Condition~\ref{con:existence} is a faithfulness type assumption on $P$, and in Theorem~\ref{ch3:thm:deterministicCorrect} we make a further assumption on $\gamma$ which ensures strong faithfulness. However, it is not strictly stronger or weaker than the
Gaussian strong faithfulness type assumption. In particular we require the linear coefficients and error moments considered to be jointly ``sufficiently parentally faithful and non-Gaussian." So for a fixed sample size, there may be cases where the linear coefficients and error covariances do not satisfy Gaussian strong faithfulness, but do satisfy the non-Gaussian condition because the higher order moments are sufficiently non-Gaussian. However, the opposite may also occur where a set of linear coefficients and error moments satisfy Gaussian strong faithfulness but not the non-Gaussian condition.

Finally, let $\mathcal{P}_{F_K}(G)$ be the subset of distributions
$P\in\mathcal{P}(G)$ with $\tau^{(K)}_{v.C\rightarrow u} \neq 0$
whenever $u \in \pa(v)$ and
$C\subseteq V\setminus\left(\{u,v\}\cup \de(v)\right)$.  Then the set
of linear coefficients and error moments % up to order $K$
that induce
an element of $\mathcal{P}(G)\setminus\mathcal{P}_{F_K}(G)$ has
measure zero.  This set difference includes distributions which are
not parentally faithful with respect to $G$ and distributions for
which there exist a parent/child pair for which both error
distributions have Gaussian moments up to order $K$.  

\begin{theorem}\label{ch3:thm:deterministicCorrect}
	For some $p$-variate distribution $P$ and data $Y = \left(Y_1, \ldots, Y_n \right)$:
	\begin{enumerate}[label={(\roman*)},ref={\thetheorem(\roman*)}]
		\item \label{imp:uniqueness}Suppose Condition~\ref{con:existence} holds.  Then among all DAGs with maximum in-degree at most $J$, there exists a unique DAG G such that $P\in P_{F_K}(G)$
		\item \label{imp:correctness}Suppose Conditions~\ref{con:existence}-\ref{con:sample_moments} hold for constants which satisfy  
		\begin{equation}
		\begin{aligned}\label{ch3:eq:gammaCond}
		\gamma/2  > \delta_3 :&= 4M \delta_1 \left\{16(3^K)(J+ K)^{K}  K \frac{J^{(K+4)/2}M^{K+1}}{\lMin^{K+1}}\right\} \\
		&\quad + 2 \left[\delta_1 \left\{16(3^K)(J+ K)^{K}  K \frac{J^{(K+4)/2}M^{K+1}}{\lMin^{K+1}}\right\}\right]^2.
		\end{aligned}
		\end{equation}
		Then with pruning parameter $g = \gamma / 2$, Algorithm~\ref{alg:topOrder} will output $\hat G = G$.
		
	\end{enumerate}
\end{theorem}

The main result of Theorem~\ref{ch3:thm:deterministicCorrect} is part
(ii). The identifiability of a DAG %from infinite data
was previously shown by \citet{shimizu2006lingam} by appealing to
results for independent component analysis; however, our direct
analysis of rational functions of $Y$ allows for an explicit tolerance
for how sample moments of $Y$ may deviate from corresponding
population moments of $P$. This implicitly allows for model
misspecification; see
% , which is addressed more explicitly in
Corollary~\ref{ch3:thm:probGuarantee}. The proof of Theorem
\ref{ch3:thm:deterministicCorrect} requires Lemmas
\ref{ch3:thm:errorBeta}-%, \ref{ch3:thm:errorMoment} and
\ref{ch3:thm:errorTau}, which we develop first.  The lemmas are proven
in  the supplement. Recall that $\beta_{vC}$ are the population regression coefficients from \eqref{eq:pop-reg-coeffs}, and let $\hat \beta_{vC}$ denote the coefficients estimated from $Y$.
\begin{lemma}\label{ch3:thm:errorBeta}
	Suppose Conditions \ref{con:min_eigen}, \ref{con:bounded_moments}, and \ref{con:sample_moments} hold. Then for any $v \in V$, $C \subseteq V$, and $|C| \leq J$, 
	\[ \|\hat \beta_{vC} - \beta_{vC}\|_\infty < \delta_2 =  4 \frac{J^{3/2}M\delta_1}{\lMin^2}.\]
\end{lemma}
% The proof of Lemma \ref{ch3:thm:errorBeta} uses well known results for
% matrix inversion.
Recall, that $Y_{vi.C} = Y_{vi} - \sum_{c \in C}\beta_{vc.C} Y_{ci}$. Let $Z_{v.C}$ denote the analogous quantity for $Z \sim P$, and let $\hat Y_{vi.C} = Y_{vi} - \sum_{c\in C}\hat \beta_{vc.C}Y_{ci}$. 
\begin{lemma}\label{ch3:thm:errorMoment}
	Suppose that Conditions \ref{con:min_eigen}, \ref{con:bounded_moments}, and \ref{con:sample_moments} hold. Let $s$, $r$ be non-negative integers such that $s + r \leq K$, and let $Z \sim P$. For any $v, u \in V$ and $C\subseteq V \setminus\{u,v\}$ such that $|C| \leq J$,
	\[\left|\frac{1}{n} \sum_i \hat Y_{vi.C}^s Y_{ui}^r - \E\left(Z_{v.C}^s Z_{u}^r\right)\right|  <\delta_1 \Phi(J,K,M, \lMin)\]
	where
	\begin{equation}
	\begin{aligned}\label{ch3:eq:phiDef}
	\Phi(J,K,M, \lMin) &= \left\{16(3^K)(J+ K)^{K}  K \frac{J^{(K+4)/2}M^{K+1}}{\lMin^{K+1}}\right\}.
	\end{aligned}
	\end{equation}
\end{lemma}
The proof of Lemma~\ref{ch3:thm:errorMoment} relies on the fact that
the map from moments of $Z$ to the quantities of interest are
Lipschitz continuous within a bounded domain.
% Finally, we deduce that the distance between $\hat \tau_{v.C \rightarrow u}$ and $\tau_{v.C \rightarrow u}$ is bounded, which we formally state in the following lemma.
\begin{lemma}\label{ch3:thm:errorTau}
	Suppose that Conditions \ref{con:min_eigen}, \ref{con:bounded_moments}, and \ref{con:sample_moments} hold. Then
	\[|\hat \tau_{v.C \rightarrow u} - \tau_{v.C \rightarrow u}| < 4M \delta_1 \Phi(J, K, M, \lMin) + 2 \left\{\delta_1 \Phi(J, K, M, \lMin)\right\}^2 = \delta_3\]
	for the function $\Phi(J, K, M, \lMin)$ given in Lemma~\ref{ch3:thm:errorMoment}.
\end{lemma}

The proof of Lemma~\ref{ch3:thm:errorTau} is an application of the
triangle inequality.
% Using Lemmas~\ref{ch3:thm:errorBeta}-\ref{ch3:thm:errorTau}, we now give a proof of Theorem~\ref{ch3:thm:deterministicCorrect}.
\begin{proof}[of Theorem~\ref{ch3:thm:deterministicCorrect}]
	(ii)
	We proceed by induction. By Lemma \ref{ch3:thm:errorTau} and assuming \eqref{ch3:eq:gammaCond}, each statistic $\hat \tau_{v.C \rightarrow u}$ is within $\delta_3 < \gamma/2$ of the corresponding population quantity. Thus, any statistic corresponding to a parameter with value 0 is less than $\gamma / 2$ and, by Condition~\ref{con:existence} and the condition on $\gamma$ in \eqref{ch3:eq:gammaCond}, all statistics corresponding to a non-zero parameter are greater than $\gamma / 2$.
	
	Recall that $\Theta^{(z)}$ is a topological ordering of
	nodes.   Assume for some step $z$, that $\Theta^{(z-1)}$ is
	consistent with a valid ordering of $G$. Let $R^{(z)} = \{v
	\in \Psi^{(z-1)}: \an(r) \subseteq \Theta^{(z-1)}\}$ so that
	any $r \in R^{(z)}$ is a root in the subgraph induced by
	$\Psi^{(z-1)}$ and $\Theta^{(z)} = (\Theta^{(z-1)} \cup \{r\})$
	is consistent with $G$.  The base case for $z = 1$ is trivially satisfied since $\Theta^{(0)} = \emptyset$.
	
	Setting $g = \gamma / 2$ does not incorrectly prune any parents, so $\pa(r) = \mathcal{C}_r^{(z)}$, which  implies for all $r \in R^{(z)}$ that $\hat T(r, \mathcal{C}_r^{(z)}, \Phi^{(z-1)}) < \gamma / 2$. Similarly, for any $v \in \Psi^{(z-1)}\setminus R^{(z)}$, there exists $u \in \Psi^{(z-1)}$ with $|\hat \tau_{v.C\rightarrow u}| > \gamma/2$ for all $C \subseteq \Theta^{(z-1)}$. Thus,
	$\hat T\big(r, \mathcal{C}^{(z)}_r, \Psi^{(z-1)}\big) < \hat T\big(v, \mathcal{C}^{(z)}_v, \Psi^{(z-1)}\big)$
	for every $r\in R^{(z)}$ and $v \in \Psi^{(z-1)}\setminus
	R^{(z)}$. This implies the next root selected, $\arg \min_{v \in \Psi^{(z-1)}} \hat T
	\big(v, \mathcal{C}^{(z)}_v, \Psi^{(z-1)}\big)$ must be in
	$R^{(z)}$, and thus $\Theta^{(z)}$ remains consistent with $G$. 
	
	(i) The fact that $P \in \mathcal{P}_{F_K}(G)$ follows
	directly from the definition. To show uniqueness, we use
	population quantities so that $\delta_1 = 0$ which in turn
	implies $\delta_3 = 0$. Then for any $\gamma > 0$,
	Algorithm~\ref{alg:topOrder} will return $G$. Thus, by
	\ref{imp:correctness}, $G$ must be unique.
\end{proof}

\begin{remark}
	As stated Theorem~\ref{ch3:thm:deterministicCorrect} concerns
an explicit cut-off $g$, whereas % in
% Section~\ref{sec:graph_recovery},
in practice we specify a tuning parameter $\alpha$ that is easier to
interpret and tune.  If $\alpha\le 1$, it holds under the conditions of
Theorem~\ref{ch3:thm:deterministicCorrect} that % if $\alpha \leq 1$,
	%, the tuning parameter described in Section \ref{sec:graph_recovery},
	Algorithm~\ref{alg:topOrder} returns a topological ordering
        consistent with $G$, but $\hat E$ may be a superset of
        $E$. However, there exists $\alpha \geq 1$ which will
        recover the exact graph.  

        To see this note that $\alpha \le 1$ ensures that
        $g^{(z)} < \gamma / 2$ under the specified conditions, so no
        parents are pruned incorrectly and the estimated
        topological ordering is correct. This, however, may
        not remove all ancestors that are not parents, so the
        estimated edge set may be a superset of the true edge
        set.  Letting instead
\begin{equation}\label{eq:optAlpha}
\alpha = \frac{\min_{v}\min_{a \in \pa(v)}\min_{C\cap \de(v) = \emptyset } \vert \hat \tau_{v.C \rightarrow a} \vert}
{ \max_{v}\max_{ a \in \an(v)\setminus \pa(v)}\min_{ C\cap \de(v) = \emptyset }\vert \hat \tau_{v.C \rightarrow a}\vert },
\end{equation}
will correctly prune ancestors and not parents. Because all sample moments are close to their population values, the denominator must be less than $\gamma/2$ and strong parental faithfulness further implies that the numerator is greater than $\gamma/2$ so \eqref{eq:optAlpha} is greater than 1. However, setting $\alpha$ too large may result in an incorrect estimate of the ordering since a true parent may be errantly pruned. Thus, we advocate a more conservative approach of setting $\alpha \leq 1$ which is more robust to violations of strong faithfulness.
\end{remark}

\begin{remark}
	Suppose $P_Y \in \mathcal{P}(G)$ but is not necessarily parentally faithful with respect to $G$. If $\alpha = 0$ and $\beta_{vu} \neq 0$ for all $(u, v) \in E$, then for generic error moments a correct ordering will still be recovered consistently as $\delta_1 \rightarrow 0$.

Indeed, Corollary~\ref{ch3:thm:aggStat}\ref{imp:aggStatNonZero} holds without parental faithfulness. So for generic error moments, there exists $\gamma > 0$ such that $T(v, \mathcal{C}_v^{(z - 1)}, \Phi^{(z-1)}) > \gamma$ for all $v \in \Phi^{(z-1)} \setminus R^{(z)}$ for all steps $z$. However, without parental faithfulness, a parent node may be errantly pruned if $\alpha > 0$.
To ensure Corollary~\ref{ch3:thm:aggStat}\ref{imp:aggStatZero} holds,
we need $\pa(r) \subseteq \mathcal{C}_v^{(z)}$ for all $r \in
R^{(z)}$, which is satisfied by letting $\mathcal{C}_r^{(z)} =
\Theta^{(z-1)}$. For fixed $\gamma$, since $\delta_3 \rightarrow 0$ as
$\delta_1 \rightarrow 0$, there exists a $\delta_1$ so that $\gamma > 2\delta_3$.
\end{remark}

\subsection{High-dimensional consistency}
We now consider a sequence of graphs, observations, and distributions
indexed by the number of variables $p$.  %: $G^{(p)}$, $Y^{(p)}$,
% $P^{(p)}_Y$, and $P^{(p)}$.
For notational brevity, we do not explicitly include the index $p$ in
the notation, and keep simply writing $G$, $Y$, $P_Y$ and $P$ for
these sequences. The following corollary states conditions sufficient
for the % YSW deterministic
conditions of Theorem~\ref{ch3:thm:deterministicCorrect} to hold with
probability tending to 1. We first make explicit assumptions on $P_Y$,
with $m^\star_{V, \alpha}$ denoting the population moments of
$P_Y$. Again, we allow for misspecification, but require control of
the $L_\infty$ distance between population moments of $P_Y$ and some
$P \in \mathcal{P}_{F_K}(G)$.

\begin{condition}\label{con:log_concave}
	$P_Y$ is a log-concave distribution.
\end{condition}
\begin{condition}\label{con:bounded_variance}
	All population moments of $P_Y$ up to degree $2K$, $m^\star_{V, \alpha}$ for $\sum_{v}\alpha_v \leq 2K$, are bounded by $M - \xi > \max(1, \lMin / J)$.
\end{condition}
\begin{condition}\label{con:misspecification}
	Each population moment of $Y$ up to degree $K$, $m^\star_{V, \alpha}$ for $\sum_{v}\alpha_v \leq K$, is within $\xi$ of the corresponding population moment of $P$.
\end{condition}
When $Y$ is actually generated from a recursive linear structural equation model, Condition~\ref{con:misspecification} trivially holds with $\xi = 0$ and log-concave errors imply that $Y$ is log-concave. 
%YSW The first condition in Corollary~\ref{ch3:thm:probGuarantee} shows how $n$ must grow relative to $p$, and the second condition shows how $\xi$ may scale relative to the other quantities.  

\begin{corollary} \label{ch3:thm:probGuarantee}
	For a sequence of distributions $P$ and data $Y$ assume Conditions~\ref{con:existence}, \ref{con:min_eigen}, \ref{con:log_concave}, \ref{con:bounded_variance}, and \ref{con:misspecification} hold. For pruning parameter $g = \gamma / 2$, Algorithm \ref{alg:topOrder} will return the graph $\hat G = G$ with probability tending to $1$ if
	\begin{gather*}
	\frac{\log(p)}{n^{1/(2K)}} \frac{J^{5/2}K^{5/2}M^2 }{\gamma^{1/2}\lMin^{3/2}} \rightarrow 0, \qquad\; \qquad
	\xi \frac{3^K K^{K + 1} J^{(3K)/2 + 2} M^{K+2}}{\gamma \lMin^{K+1}} \rightarrow 0 \numberthis
	\end{gather*}
	when $p \rightarrow \infty $ and $\gamma, \lMin < 1 < M$.
\end{corollary}

\begin{proof}
	Conditions~\ref{con:bounded_variance} and \ref{con:misspecification} imply Condition \ref{con:bounded_moments}. It remains to be shown that Condition~\ref{con:sample_moments} and \eqref{ch3:eq:gammaCond} hold for the $\gamma$ specified in Condition~\ref{con:existence}. Solving the inequality in Lemma~\ref{ch3:thm:errorTau} for $\delta_1$ shows \eqref{ch3:eq:gammaCond} will be satisfied if the sample moments of $Y$ are within $\delta$ of the population moments %of $Y$ for some $\delta$ 
	such that $\delta + \xi \leq \delta_1$ with $\delta_1$ less than
	\begin{equation*}
	\begin{aligned}
	\min\left[\frac{-8M \Phi + \left\{(8M\Phi)^2 + 16\Phi^2\gamma\right\}^{1/2}}{8\Phi^2}, \frac{\lMin}{2J}, M\right]  & = \min\left\{\frac{\left(M^2 + \gamma / 4\right)^{1/2} - M}{\Phi}, \frac{\lMin}{2J}\right\} \\
	\end{aligned}
	\end{equation*}
	for $\Phi$ defined in \eqref{ch3:eq:phiDef}. Since $J, K, M > 1$, $\gamma, \lMin < 1$ ensure that first term is the relevant term. We further simplify the expression since
	\[\left(M^2 + \gamma / 4\right)^{1/2} \geq M + \gamma \min_{t\in (0, \gamma)} \left.\frac{\partial \left(M^2 + \gamma / 4\right)^{1/2}}{\partial \gamma}\right|_{\gamma = t} = M + \frac{\gamma}{8 \left(M^2 + \gamma / 4\right)^{1/2}}. \]
	Thus, the conditions of Theorem~\ref{ch3:thm:deterministicCorrect} will be satisfied if
	\begin{align*}
	\delta + \xi  \leq \frac{\gamma}{8\left(M^2 + \gamma / 4\right)^{1/2}\Phi} =: \delta_4.
	\end{align*}
	Specifically, we analyze the case when $\xi < \delta_4/2$ and $|\hat m_{V,a} - m_{V,a}| < \delta < \delta_4/2$ for all $|a| \leq K$.
	If $Y_v$ follows a log-concave distribution, we can apply Lemma B.3 of \citet{lin2016estimation} which states for
	$f$, some $K$ degree polynomial of log-concave random variables $Y = (Y_1, \ldots, Y_n)$, and some absolute constant, $L$, if
	\begin{equation*}
	\frac{2}{L} \left(\frac{\delta}{(e)\left[\var\left\{ f(Y)\right\}\right]^{1/2}} \right)^{1/K} \geq 2
	\end{equation*} then
	\begin{equation*}
	{\rm Pr}\left[|f(Y) - \E\left\{f(Y)\right\}| > \delta \right] \leq \exp\left\{\frac{-2}{L} \left(\frac{\delta}{\left[\var\left\{ f(Y)\right\}\right]^{1/2}} \right)^{1/K}\right\}.
	\end{equation*}
	
	Letting $f(Y)$ be the sample moments of $Y$ up to degree $K$, Condition~\ref{con:bounded_variance} implies the variance is bounded by $M / n$. When $p > 2$, there are $\binom{p + K }{p} < p^{K}$ moments with degree at most $K$, then by a union bound, when $0 <\xi < \delta_4/2$,
	\begin{equation*}
	\begin{aligned}
	{\rm Pr}\left(\hat G = G\right) &\geq 1 - {\rm Pr}\left(|\hat
          m_{V,a} - m_{V,a}| > \delta_4/2\text{ for any } |a| \leq K
        \right) \\
	&\geq 1 - p^K \exp\left[\frac{-2}{L} \left\{\frac{ \delta_4/2}{\left(M/n\right)^{1/2}} \right\}^{1/K}\right]
	\end{aligned}
	\end{equation*}
	when 
	\begin{equation}\label{eq:sampleSize}
	\frac{2n^{1/(2K)}}{L} \left(\frac{ \delta_4/2}{eM^{1/2}} \right)^{1/K} \geq 2.
	\end{equation}
	In the asymptotic regime, where $p$ is increasing,
	\begin{equation*}
	\begin{aligned}
	\frac{LM^{1/(2K)}K\log(p)}{\left( \delta_4/2\right)^{1/K}n^{1/(2K)}} \rightarrow 0
	\end{aligned}
	\end{equation*}
	implies that the inequality in \eqref{eq:sampleSize} will be satisfied and
	\[ p^K \exp\left[\frac{-2}{L} \left\{\frac{\delta_4/2 }{\left(M/n\right)^{1/2}} \right\}^{1/K}\right] \rightarrow 0. \]
	Plugging in the expression for $\delta_4$, we find
	\begin{align*}
	\frac{LM^{1/(2K)}K\log(p)}{\left( \delta_4/2\right)^{1/K}2n^{1/(2K)}} & = \frac{LM^{1/(2K)}K\log(p)}{2n^{1/(2K)}} \times \\
	&\quad \left\{\frac{16\left(M^2 + \gamma / 4\right)^{1/2} 16 (3^K)(J+K)^{K}K J^{(K+4)/2}M^{K+1} }{\gamma\lMin^{K+1}}\right\}^{1/K}.
	\end{align*}
	This quantity is of order $\mathcal{O}\big(
	\big(\log(p)J^{5/2}K^{5/2}M^2\big)/\big(n^{1/(2K)}
	\gamma^{1/2}\lMin^{3/2}\big)\big)$ when assuming that $\gamma < M$.
	In addition, $\xi < \delta_4 /2$ will be satisfied if
	$\frac{2\xi}{\delta_4} \rightarrow 0$. This ratio is
	\begin{align*}
	\frac{2\xi}{\delta_4} &= 2\xi \left\{\frac{16\left(M^2 + \gamma / 4\right)^{1/2} 16 (3^K)(J+K)^{K}K J^{(K+4)/2}M^{K+1} }{\gamma\lMin^{K+1}}\right\}
	\end{align*}
	which is $\mathcal{O}\big(\left(\xi 3^K K^{K + 1} J^{(3K)/2 + 2} M^{K+2}\right)/ \big(\gamma \lMin^{K+1}\big)\big)$ when $\gamma < M$.
\end{proof}

When fixing the other terms, Corollary~\ref{ch3:thm:probGuarantee}
requires $\log(p) = o(n^{1/(2K)})$. % which allows for consistency
% even when $p > n$.
Corollary~\ref{ch3:thm:probGuarantee} does not preclude $J$ from growing with $n$ and $p$; however, the computational complexity of Algorithm~\ref{alg:topOrder} is exponential in $J$, so in practice $J$ must remain relatively small. 

\section{Numerical results}\label{highD:sec:simulations}
\subsection{Simulations: low dimensional performance}\label{sec:sim_min_max}
We first compare the proposed method using: (1) min-max $\hat T_1$ and (2) max-min $\hat T_2$ against (3) DirectLiNGAM \citep{shimizu2011direct} and (4) Pairwise LiNGAM \citep[Section 3.2]{hyvarinen2013pairwise}.  We randomly generate graphs and corresponding data with the following procedure. For each node $v$, % = 2, \ldots, p$,
select the number of parents $d_v$ uniformly from $1, \ldots, \min(v, J)$. We %YSW always 
include edge $(v-1, v)$ to ensure that the ordering is unique and draw $\beta_{v, v-1}$ uniformly from $(-1, -.5) \cup (.5, 1)$. The remaining %YSW $d_v - 1$  
parents are selected uniformly from $[v-2]$ and the corresponding edge weights are set to $\pm 1/5$. The $n$ error terms for variable $v$ are generated by %YSW first drawing a standard deviation 
selecting $\sigma_v \sim \text{unif}(.8, 1)$ and then drawing $\varepsilon_{vi} \sim  \sigma_v\text{unif}(-\sqrt{3}, \sqrt{3})$.

We use $K = 4$, fix the max in-degree $J = 3$, let $p = 5, 10, 15,
20$, and let $n = 50p$ and $n = 10p$.  %Since $\gamma/2$ is not known,
We set $\alpha = .8$ % to ensure that at least the topological
% ordering can be recovered consistently. We
and compare performance by measuring Kendall's $\tau$ between the returned ordering and the true ordering; i.e., the number of concordant pairs in the ordering minus the number of discordant pairs, normalized by the number of total pairs. The procedure is repeated 500 times for each setting of $p$ and $n$.

\begin{figure}[t]
	\centering
	\includegraphics[scale = .6]{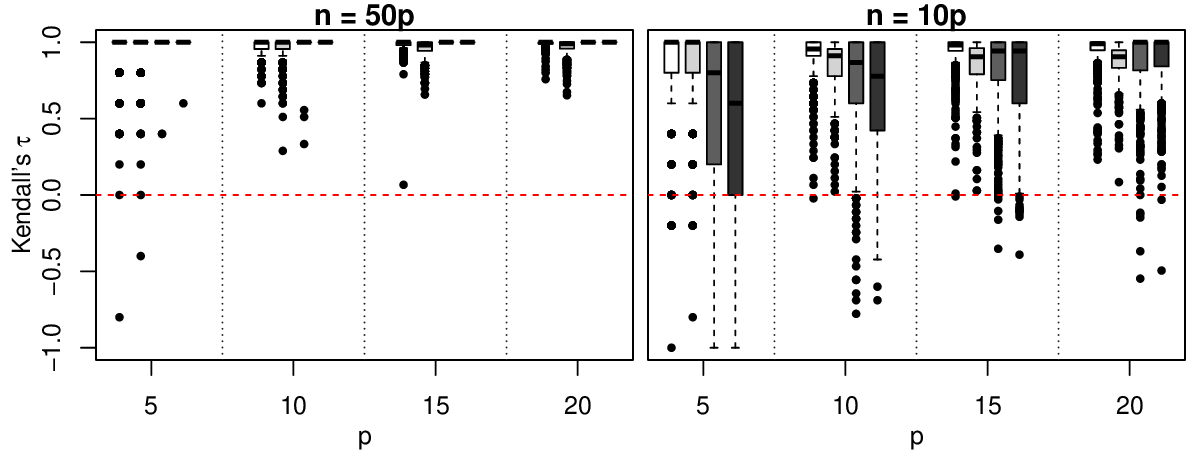}
	\caption[Low dimensional DAG estimation comparison]{\label{fig:comparison}Each bar represents the results from 500 randomly drawn graphs and data. In each group, from left to right, %YSW (lightest to darkest),
		the bars represent (1) min-max $\hat T_1$, (2) max-min $\hat T_2$, (3) % DirectLiNGAM
		\citet{shimizu2011direct}, and (4) %YSW Pairwise LiNGAM
		\citet{hyvarinen2013pairwise}. In the left panel $n = 50p$ and the right panel $n = 10p$.}
\end{figure}

%YSW The simulation results are shown in Figure~\ref{fig:comparison}. 
Figure~\ref{fig:comparison} shows that in the low-dimensional case
with $n = 50p$, the Pairwise LiNGAM and DirectLiNGAM methods
outperform the proposed method, with either statistic. However,
% in the medium-dimensional case
already
with $n = 10p$, our method begins to
give improvements.  The min-max statistic $T_1$ does slightly better
than $T_2$, the max-min.  However,
Figure~\ref{fig:timing} %YSW we observe
shows a large difference in computational effort;
% we also
% include %YSW the cases where
$p = 40, 80$ are included for further contrast. In the sequel, we
use the max-min statistic, $T_2$.
\begin{figure}[t]
	\centering
	\includegraphics[scale = .5]{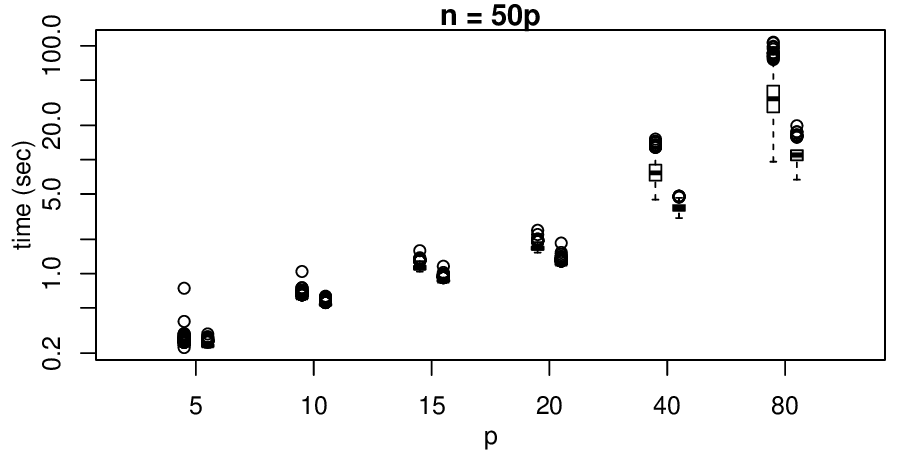}
	\caption[Timing comparison of min-max vs max-min
	methods]{\label{fig:timing} Timing results from 500 randomly
		drawn graphs and data with $n = 50p$. In each pair, the left
		represents min-max, $\hat T_1$ and the right max-min,
		$\hat T_2$. The y-axis is on a log scale.}
\end{figure}

The proposed method compares favorably to the DirectLiNGAM method in
computational effort because of the expensive kernel mutual
information calculation and is comparable to the Pairwise
LiNGAM. However, we refrain from a direct timing comparison because
DirectLiNGAM and Pairwise LiNGAM are both implemented in Matlab while
our proposed method is implemented in R and C++
\citep{r2017,eddelbuettel2011rcpp}. In the supplement, we also provide
a direct comparison between the proposed statistic and those used by
\citet{shimizu2011direct} and \citet{hyvarinen2013pairwise}.

\subsection{Simulations: high-dimensional consistency}\label{sec:highDSims}
To illustrate high-dimensional consistency, we generate the graph and
coefficients as in Section \ref{sec:sim_min_max} but with $p = 100, 200, 500, 1000, 1500, 2000$ and $n = 3/4p$. We first consider random DAGs and data generated as before, but with $J = 2$. We also consider %YSW the case where the graph may contain hub nodes
graphs with hubs, that is, nodes with large out-degree.
%% MD:  somewhat repetitive, we say this in intro already so I cut
%% here
% In our analysis, we only make assumptions about the in-degree, while the PC algorithm %YSW under Gaussian errors
% requires the maximum total degree, both in and out, to grow sub-linearly \citep{kalisch2007estimating}. %YSW However, in gene regulatory networks, there are often hub nodes which %YSW have a very high out-degree and  regulate many downstream genes \citep{hao2012revisiting}. 
%% We generate random graphs with hubs %YSW using the following procedure. We 
These are generated by including a directed edge from $v-1$ to $v$ for all nodes $v = 2, \ldots, p$ and drawing the edge weight uniformly from $(-1, -.65)\cup (.65, 1)$. %YSW The standard deviations for each of the error terms is drawn uniformly from $(.8, 1)$.
We then set nodes $\{1,2,3\}$ as hubs and include an edge with weight $\pm 1/5$ to each non-hub node from a randomly selected hub. Thus, the out-degree for each of the hub nodes grows linearly with $p$, but the maximum in-degree remains bounded by $2$. %YSW Again, we let $p = 100, 200, 500, 1000, 1500$ and $n =3/4 p$. 
For both %YSW the general graph and hub graph 
cases, the results for 20 runs at each value of $p$ are shown in Figure~\ref{fig:highDCons}.
\begin{figure}[t]
	\centering
	\includegraphics[scale = .5]{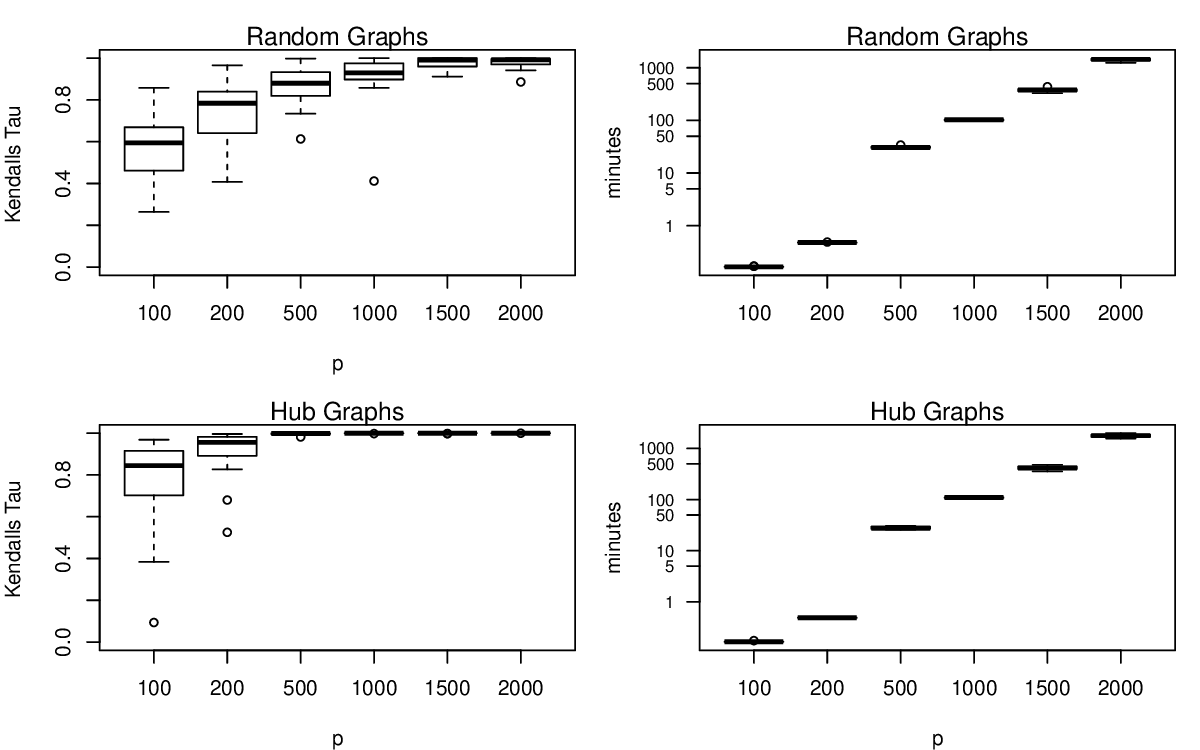}
	\caption[High-dimensional consistency of HDL]{\label{fig:highDCons}Each boxplot represents the results of 20 simulations. In all cases, we let $n = 3/4p$. The top panels show results from randomly drawn DAGs while the bottom panel shows results from DAGs constructed to have hub structure. The left plots show performance as measured by Kendall's $\tau$ and the right plots show computational time when using 16 CPUs in parallel.}
\end{figure}
In the supplement, we show simulations with gamma errors and also consider a setting with Gaussian errors, where our method should not be consistent.

\subsection{Pre-selection of neighborhoods}\label{sec:preSelectSims}
As with the original DirectLiNGAM procedure, any edges or non-edges
known in advance can be accounted for.
% incorporated into the search procedure.
Such information could, for instance, be obtained by
applying
% In particular, we may use
neighborhood selection \citep{meinshausen2006neighborhood} to estimate
the Markov blanket of each node.  This blanket consists of parents,
children, and parents of children.  For sparse graphs, \citet[Section
3.3]{hyvarinen2013pairwise} propose first using such a pre-selection step, then directly estimating the direction of each edge using pairwise measures without any additional adjustment. To create a total ordering, Alg B and Alg C of \citet{shimizu2006lingam} can be used. This does not require specifying a maximum in-degree, but in general, the neighborhood selection procedure will only be consistent if the total degree is controlled. 

In our proposed procedure, we may incorporate estimated Markov
blankets by limiting, at each step $z$, for each remaining node $v$,
the set of potential parents, $C_v^{(z)}$, to the intersection of the
estimated Markov blanket of $v$ and the previously ordered nodes,
$\Theta^{(z-1)}$. We do not otherwise prune the set of potential
parents. Figure~\ref{fig:preSelect} shows results from using the
pre-selection step under the setting from
Section~\ref{sec:highDSims} for general random graphs. The
pre-selection procedure improves the performance of our proposed
high-dimensional LiNGAM procedure, but the proposed procedure without
pre-selection still outperforms the two-stage procedure of
\citet[Section 3.3]{hyvarinen2013pairwise}.  Similar results for the hub graph
setting are shown in the supplement.  % and are qualitatively similar.

\begin{figure}[t]
	\centering
	\includegraphics[scale = .4]{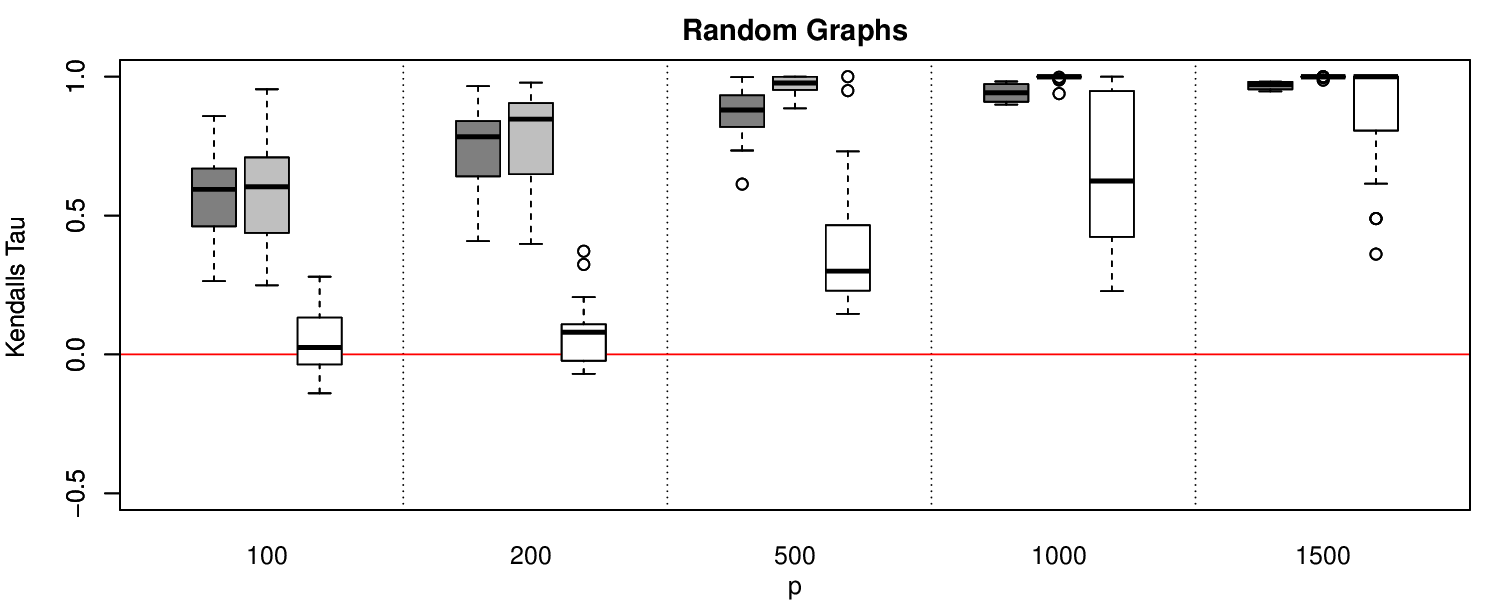}
	\caption{\label{fig:preSelect}Each boxplot represents 20 simulations with random DAGs when using a pre-selection step; in each case $n = 3/4p$. From left to right the methods are: the proposed high-dimensional LiNGAM procedure, same as Figure~\ref{fig:highDCons}; the proposed high-dimensional LiNGAM procedure with pre-selection; the two stage pairwise procedure from \citet{hyvarinen2013pairwise}.}
\end{figure}  

\subsection{Data example: high-dimensional performance}
We % use the proposed %YSW high-dimensional LiNGAM 
% method to
estimate causal structure among the stocks in the Standard
and Poor's 500. Specifically, we consider the percentage
increase/decrease for each share price for each trading day between
Jan 2007 to Sep 2017. We consider the $p=442$ companies for which data is
available for the entire period, and   we scale and center the data so that each variable has mean 0 and variance 1. As %YSW we believe 
structure may vary over time,
% and the observations may only
% be approximately identically distributed for short periods, %the only
                                %data available is observational and
                                %high-dimensional. We
we estimate the causal structure for each of the following periods separately with $J=3$ and $K = 4$: 2007-2009, 2010-2011, 2012-2013, 2014-2015, 2016-2017 (ending in September). Across these periods, the sample size, $n$, ranges from 425 to 755. 

The underlying structure is unlikely to be causally sufficient or acyclic. In addition, although it is common to assume that daily returns are independent, this assumption may not hold in practice. Nonetheless, %YSW when using the sector of each individual company to assess the estimated causal structure, we still find that 
the method still recovers reasonable structure. We first consider the most recent Jan 2016 - Sep 2017 period. %In the estimated ordering, 160 of the companies directly follow another company in the same sector. When arranging the companies randomly, the probability of 160 or more matches is essentially 0. 
Figure~\ref{fig:ch3:spyOrdering} shows a boxplot for the estimated ordering of the companies within each sector. The sectors are sorted top to bottom by median ordering. %YSW (top is cause, bottom is effect). 
Near the top, we see utilities, energy, real estate, and finance. Since energy is an input for almost every other sector, intuitively price movements in energy should be causally upstream of other sectors. The estimated ordering of utilities might seem surprising; %YSW since utility companies typically trade with very little volatility; 
however, utility stocks are typically thought of as a proxy for bond prices. Thus, the estimated ordering may reflect the fact that changes in utility stocks capture much of the causal effect of interest rates, which had stayed constant for much of 2011-2015 but began moving again in 2016. Real estate and finance, sectors that are highly impacted by interest rates, are also estimated to be early in the causal ordering.

Figure~\ref{fig:ch3:orderingByTime} ranks each sector by the median
topological ordering for each period. The orderings are relatively stable over time, but there are a few notable changes. %YSW In particular, 
In 2007, real estate was estimated to be the ``root sector" while finance is in the middle. This aligns with the idea that the root of the 2008 financial crisis was actually %YSW caused by 
failing mortgage backed securities in real estate, which had a causal effect on finance. However, over time, real estate has moved more downstream.

\section{Discussion}
We proposed a causal discovery method that was proven consistent for
specific test statistics and log concave errors.  Similar analyses
could be given for other  statistics that are Lipschitz
continuous in the sample moments over a bounded domain, can
distinguish causal direction, and indicate the presence of
confounding. This would include a normalized version of the proposed
test statistics which accounts for the scaling of the
data. Log-concavity was assumed for exponential concentration of
sample moments and other distributional assumptions could be
considered instead if analogous concentration results can be obtained and
traced throughout the analysis.
% In
% addition, the result would apply directly to other classes of
% distributions for which the sample moments concentrate at an
% exponential rate.

The proposed algorithm requires selecting a bound on the in-degree $J$
and a pruning parameter $\alpha$. The in-degree is typically unknown,
but a reasonable upper bound may be used as a ``bet on
sparsity''. If the maximum in-degree of the  true graph is
larger than the specified $J$ but the ``extra edges" have small enough
edge-weights, the ``closest" DAG with maximum  in-degree $J$ is
still  recovered with high probability. The pruning parameter
$\alpha$ plays a similar role to the nominal level for each
conditional independence test in the PC algorithm. Both parameters
have an effect on the sparsity of the estimated graph and
regulate the maximum size of conditioning sets.

\begin{figure}
	\centering
	\includegraphics[scale = .55]{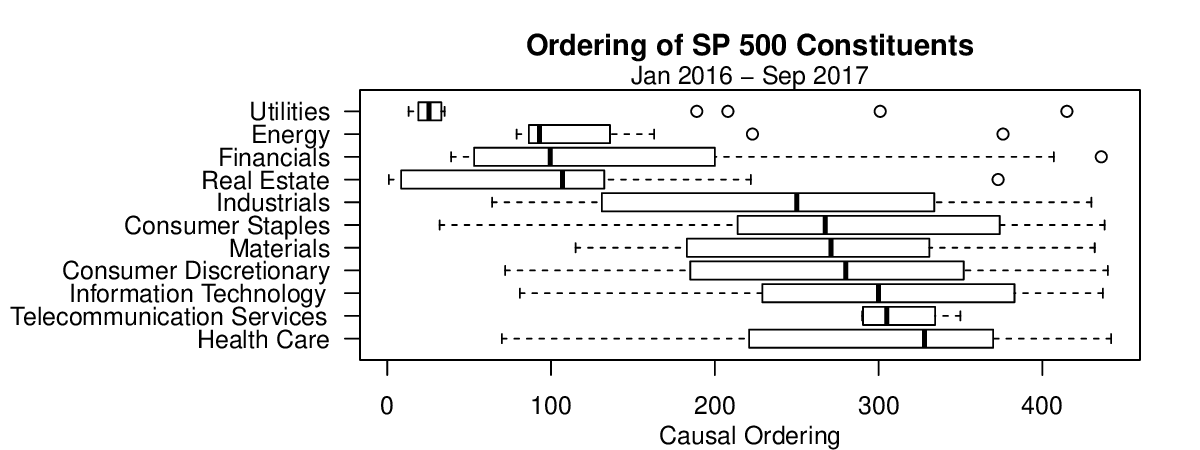}
	\caption[Data example: estimated ordering of S\&P 500 by sector; 2016 - 2017]{\label{fig:ch3:spyOrdering}Estimated causal ordering of the stocks in the Standard and Poor's 500 for Jan 2016 - Sep 2017. The stocks are grouped by sector, and the sectors are arranged by median causal ordering.}
\end{figure}  

\begin{figure}
	\centering
	\includegraphics[scale = .55]{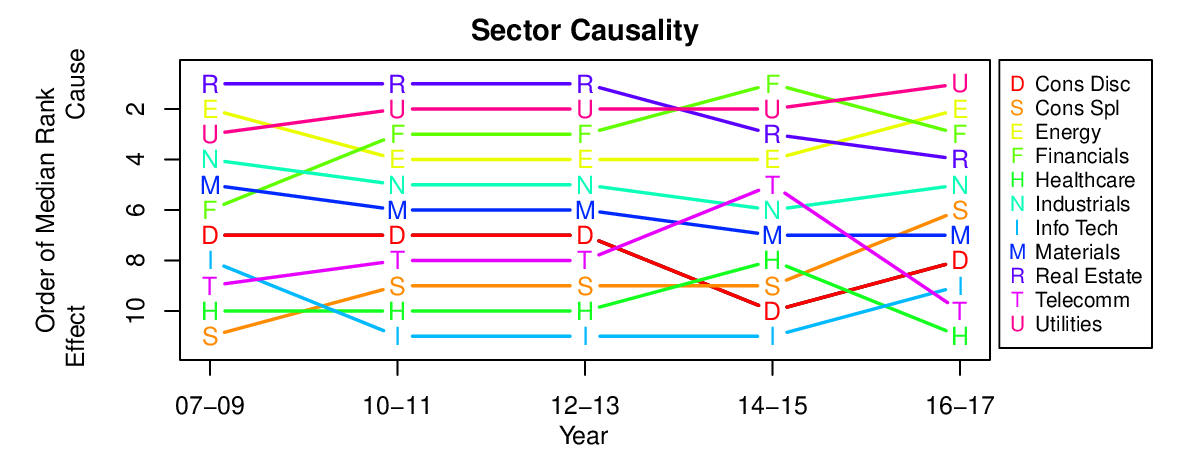}
	\caption[Data example: estimated ordering of S\&P 500 by sector; 2007 - 2017]{\label{fig:ch3:orderingByTime}Sectors ranked by median estimated topological ordering across each time period.}
\end{figure}  

At each step, instead of taking the minimum $|\tau|$ over all subsets
of potential parents, one could also pick parents for every unordered
node using a variable selection procedure and then only calculate
$|\tau|$ using the selected parents.  Such a procedure would also
consistently estimate the causal ordering as long as the variable
selection procedure is consistent.  Slightly different conditions,
such as a beta-min condition, would be needed when adopting standard
methods based on least squares,
% (i.e., second moments instead of the third moments we treated),
but in practice the resulting method performs quite well as shown in
simulations in the supplement.  This could be explained as being due
to use of only second moments for the variable selection.

In \eqref{ch3:eq:gammaCond}, we have made a key restriction that the
error moments must be adequately different from the moments of any
Gaussian and the edge weights must be strongly parentally faithful. In
practice, this is a difficult condition to satisfy, and
\citet{uhler2013geometry} show that strong faithfulness type restrictions can be problematic in practice. However, even if the distribution is not strongly parentally faithful, we can still consistently recover the correct ordering as long as each individual linear coefficient is non-zero and the errors are sufficiently non-Gaussian.
\citet{sokol2014quantifying} consider identifiability of independent component analysis for fixed $p$ when the error terms are Gaussians contaminated with non-Gaussian noise. In particular, when the effect of the non-Gaussian contamination decreases at an adequately slow rate, the entire mixing matrix is identifiable asymptotically. In our analysis, the measure of non-Gaussianity is treated by our assumptions on $\gamma$. Our results suggest that the results of \citet{sokol2014quantifying} can also be extended, given suitable sparsity, to the asymptotic regime where the number of variables is increasing.

The modified procedure we propose retains the existing benefits of the original DirectLiNGAM procedure. In particular, the output of algorithm is independent of the ordering of the variables in the input data. Although this is typically not an issue in the low-dimensional case, in the high-dimensional setting, the output of causal discovery methods may be highly dependent on ordering \citep{colombo2014order}. 
%In addition, as discussed, any edges (or non-edges) known in advance can be easily incorporated into the search procedure. \citet{loh2014inverse} show that the precision matrix recovers the moralized graph even with non-Gaussian errors. This could potentially be used as a pre-processing step to greatly decrease computational effort.

\section*{Acknowledgment}
This work was supported by the U.S.~National Science Foundation under Grant No.~DMS 1712535. 
 Thomas S. Richardson gave helpful feedback on an advance copy of the manuscript.  

\bibliographystyle{apalike}
\bibliography{diss_bib}

\appendix

\section{Proof of Theorem~\ref{ch3:thm:testStat}}\label{sec:proof_thm1}
\begin{proof}
	\textbf{Statement (i):} Consider any set $C$ such that $\pa(v)\subseteq C$ and $C \cap\de(v)= \emptyset$. Since we condition on all parents, $\beta_{vc.C} = 0$ for any $c \in C$ which is not a parent of $v$. Then,
	\begin{align*}\numberthis
	Y_{v.C} &= Y_v - \sum_{c \in \pa(v)}\beta_{vc.C}Y_{c} -\sum_{k \in C \setminus \pa(v)}\beta_{vc.C}Y_{c}\\
	& = Y_v - \sum_{c \in \pa(v)}\beta_{vc.C}Y_{c}\\
	& = \varepsilon_v.
	\end{align*}
	We then directly calculate the parameter for the set $C$:
	\begin{align*}\numberthis
	\tau^{(K)}_{v.C\rightarrow u} &= \E(Y_{v.C}^{K-1}Y_u) \E(Y_{v.C}^2) - \E(Y_{v.C}^K) \E(Y_{v.C} Y_u)
	\\
	&=  \textstyle  \E\left\{\varepsilon_v^{K-1} \left(\varepsilon_u + \pi_{uv}\varepsilon_v + \sum_{z \in \an(u)\setminus\{v\}} \pi_{uz}\varepsilon_z \right)\right\}
	\E\left(\varepsilon_v^2\right)
	\\
	& \textstyle \quad - \E\left(\varepsilon_v^K\right) \E\left\{\varepsilon_v \left(\varepsilon_u + \pi_{uv}\varepsilon_v + \sum_{z \in \an(u)\setminus\{v\}} \pi_{uz}\varepsilon_z \right)\right\}
	\\
	&=  \pi_{uv}\E\left(\varepsilon_v^K \right)
	\E\left(\varepsilon_v^2\right) - \pi_{uv}\E\left(\varepsilon_v^K\right) \E\left(\varepsilon_v^2\right)
	\\
	& = 0. 
	\end{align*}
	The penultimate equality follows from the assumption of independent errors.\qed
	\vspace{.5cm}
	
	\textbf{Statement (ii):} For fixed $C$, $u, v \in V$, and parentally
	faithful linear coefficients and variances,
	$\tau^{(K)}_{v.C\rightarrow u}$ is a polynomial of the error moments
	of degree $k = 3, \ldots, K$. Thus, selecting a single point (of
	error moments) where the quantity $\tau_{v.C \rightarrow u}$ is
	non-zero is sufficient for showing that the quantity is non-zero for
	generic error moments of degree $k = 3, \ldots, K$
	\citep{okamoto1973distinctness}. Specifically, we select that point
	by letting all error moments for $k < K$ be consistent with the
	Gaussian moments implied by $\sigma_v^2$, the variance of
	$\varepsilon_v$, but select the $K$th degree moment to be
	inconsistent with the corresponding Gaussian moment.  Since there
	are a finite number of sets $C \subseteq V$ such that
	$C \cap\de(v) = \emptyset$, then the set of error moments of degree
	$k = 3, \ldots,K$ which yield $\tau_{v.C \rightarrow u} = 0$ for any
	$C \subseteq V$ also has Lebesgue measure zero.
	
	Recall that the total residual effect of $u$ on $v$ given $C$ is
	\[
	\pi_{vu.C} = \pi_{vu} - \sum_{c \in C}\beta_{vc.C}\pi_{cu},
	\]
	where $\pi_{vu}$ is the total effect of $u$ on $v$ and
	$\pi_{uu} = 1$.  Now,
	\begin{align*} \numberthis
	Y_{v.C} &= Y_v - \sum_{c \in C}\beta_{vc.C}Y_c\\
	&= \varepsilon_v + \sum_{k \in \an(v)}\pi_{vk}\varepsilon_k - \sum_{c \in C}\beta_{vc.C}\sum_{d \in \An(c)}\pi_{cd}\varepsilon_d \\
	&= \varepsilon_v + \sum_{k \in \an(v)\cup \An(C)} \pi_{vk.C} \varepsilon_k.
	\end{align*}
	By the parental faithfulness assumption, since $u \in \pa(v)$, $\pi_{vu.C} \neq 0$.  We partition $\An(v)\cup \An(C)$ into three disjoint sets	
	\begin{align*}\numberthis
	Z_1 &= \left\{\An(v)\cup \An(C) \right\} \setminus \An(u), \\
	Z_2 &= \{z \in \An(u): \pi_{vz.C}= \pi_{vu.C}\pi_{uz}\}, \text{ and }\\
	Z_3 &= \An(u) \setminus Z_2.
	\end{align*}
	For generic edge weights, $Z_2\setminus \{u\}$
	corresponds to ancestors of $u$ which only have
	directed paths to $v$ through $C \cup\{u\}$; $Z_3$
	corresponds to ancestors of $u$ which have directed
	paths to $v$ that do not pass through $C
	\cup\{u\}$. For some specific edge weight values,
	$Z_2$ may also include non ancestors of $u$ which have
	directed paths to $v$ that do not pass through
	$C \cup\{u\}$ if those paths do not contribute to the
	total effect; i.e., non-faithfulness may result in
	$Z_2$ including more nodes. However, $Z_2$ is
	non-empty because $u\in Z_2$ by the parental
	faithfulness assumption.
	
	Let
	\begin{align*}
	\varepsilon_{Z_1} &=  \varepsilon_v + \sum_{z\in Z_1} \pi_{vz.C}\varepsilon_z,\\
	\varepsilon_{Z_2} &= \sum_{z \in Z_2} \pi_{uz}\varepsilon_z, \\
	\varepsilon_{Z_3} &= \sum_{z \in Z_3}\bigg(\pi_{vz} - \sum_{c \in C}\beta_{vc.C}\pi_{cz}\bigg) \varepsilon_z = \sum_{z \in Z_3}\pi_{vz.C} \varepsilon_z  ,  
	\end{align*}
	so that
	\[Y_{v.C} =  \varepsilon_{Z_1} + \pi_{vu.C}\varepsilon_{Z_2} + \varepsilon_{Z_3} \]
	and
	\[Y_u = \varepsilon_{Z_2}  + \sum_{z \in Z_3} \pi_{uz}\varepsilon_z. \]	
	For $w = (w_1, w_2, w_3)$ with $|w| = w_0$, let $\binom{w_0}{w} = \frac{w_0!}{w_1! w_2! w_3!}$, the multinomial coefficient. Then
	\begin{align*}\numberthis
	Y_{v.C}^{w_0} &= \left(  \varepsilon_{Z_1} + \pi_{vu.C}\varepsilon_{Z_2} + \varepsilon_{Z_3} \right)^{w_0} =  \sum_{|w| = w_0}\binom{w_0}{w}  \varepsilon_{Z_1}^{w_1} (\pi_{vu.C}\varepsilon_{Z_2})^{w_2}  \varepsilon_{Z_3}^{w_3}\\
	&= \sum_{\substack{|w| = w_0\\ w_3 = 0}}{w_0 \choose w}  \varepsilon_{Z_1}^{w_1} (\pi_{vu.C} \varepsilon_{Z_2})^{w_2}  \varepsilon_{Z_3}^{w_3} + \sum_{\substack{|w| = w_0\\ w_3 > 0}}{w_0
		\choose w}  \varepsilon_{Z_1}^{w_1} (\pi_{vu.C}  \varepsilon_{Z_2})^{w_2}  \varepsilon_{Z_3}^{w_3}\\
	&= ( \varepsilon_{Z_1} + \pi_{vu.C} \varepsilon_{Z_2})^{w_0} + \sum_{\substack{|w| = w_0\\ w_3 > 0}}\binom{w_0}{w}  \varepsilon_{Z_1}^{w_1} (\pi_{vu.C}  \varepsilon_{Z_2})^{w_2}  \varepsilon_{Z_3}^{w_3},
	\end{align*}
	so that
	\begin{align*}\numberthis\label{eq:Total}
	\tau^{(K)}_{v.C\rightarrow u} & = \E (Y_{v.C}^{K-1} Y_u) \E (Y_{v.C}^2) - \E (Y_{v.C}^K) \E (Y_{v.C} Y_u)
	\\
	& = \E\left\{\left(  \varepsilon_{Z_1} + \pi_{vu.C} \varepsilon_{Z_2} + \varepsilon_{Z_3} \right)^{K-1} Y_u \right\}
	\E\left\{ \left( \varepsilon_{Z_1} + \pi_{vu.C}\varepsilon_{Z_2} + \varepsilon_{Z_3}\right) ^2\right\}\\
	&\quad - \E\left\{\left(  \varepsilon_{Z_1} + \pi_{vu.C}\varepsilon_{Z_2} + \varepsilon_{Z_3} \right)^K\right\} \E\left\{\left(  \varepsilon_{Z_1} +\pi_{vu.C} \varepsilon_{Z_2} + \varepsilon_{Z_3} \right) Y_u\right\}
	\\
	& = \left(\sigma_{Z_1}^2 + \pi_{vu.C}^2\sigma_{Z_2}^2 + \sigma_{Z_3}^2\right)\left[\E\left\{( \varepsilon_{Z_1} + \pi_{vu.C}\varepsilon_{Z_2})^{K-1}Y_u\right\}  \vphantom{\E\left\{Y_u\sum_{\substack{|w| = K-1\\ w_3 > 0}}\binom{K-1}{w}  \varepsilon_{Z_1}^{w_1} (\pi_{vu.C}  \varepsilon_{Z_2})^{w_2}  \varepsilon_{Z_3}^{w_3} \right\}}\right.\\
	&\left. \hphantom{\left(\sigma_{Z_1}^2 + \pi_{vu.C}^2\sigma_{Z_2}^2 + \sigma_{Z_3}^2\right)} \quad \quad  + \E\left\{Y_u\sum_{\substack{|w| = K-1\\ w_3 > 0}}\binom{K-1}{w}  \varepsilon_{Z_1}^{w_1} (\pi_{vu.C}  \varepsilon_{Z_2})^{w_2}  \varepsilon_{Z_3}^{w_3} \right\}\right]\\
	&\quad - \left[\E\left\{( \varepsilon_{Z_1} + \pi_{vu.C}\varepsilon_{Z_2})^{K}\right\} + \E\left\{\sum_{\substack{|w| = K\\ w_3 > 0}}{K \choose w}  \varepsilon_{Z_1}^{w_1} (\pi_{vu.C}  \varepsilon_{Z_2})^{w_2}  \varepsilon_{Z_3}^{w_3}\right\}\right] \\
	&\quad \times
	\left\{\pi_{vu.C}\sigma_{Z_2}^2 + \E\left( \varepsilon_{Z_3} Y_u\right)\right\}.
	\end{align*}

	We first consider the case where $Z_3$ is empty so that the expansion above reduces to
	\begin{align*}\numberthis
	\tau^{(K)}_{v.C\rightarrow u} & = \left[\E\left\{( \varepsilon_{Z_1} + \pi_{vu.C}\varepsilon_{Z_2})^{K-1}Y_u\right\}\right]
	\times
	\left(\sigma_{Z_1}^2 + \pi_{vu.C}^2\sigma_{Z_2}^2\right)\\
	&\quad - \left[\E\left\{( \varepsilon_{Z_1} + \pi_{vu.C}\varepsilon_{Z_2})^{k}\right\}\right] \times
	\left(\pi_{vu.C}\sigma_{Z_2}^2\right).
	\end{align*}
	Let the moments of degree $k$ for $2 < k <  K$ of all the error terms be consistent with some Gaussian distribution. This implies that the error moments for $z = Z_1, Z_2$ are also consistent with some Gaussian distribution since the sum of Gaussians is also Gaussian. So $\E(\varepsilon_z^2) = \sigma_z^2$ and for $k < K$,
	\[E(\varepsilon_z^k) =
	\begin{cases}
	0 & \text{ if } k \text{ is odd},\\
	(k-1)!!\sigma_z^k &\text{ if } k \text{ is even},
	\end{cases}\]
	where $k!!$ is the double factorial of $k$. However, let the $K$th degree error moments be inconsistent with the specified Gaussian distribution so that for $z = Z_1, Z_2$ and $\eta_z > 0$,
	\[E(\varepsilon_z^K) =
	\begin{cases}
	\eta_z & \text{ if }  K \text{ is odd},\\
	(K-1)!!\sigma_z^K + \eta_z &\text{ if } K \text{ is even.}\\
	\end{cases}\]
	By direct calculation we see
	\begin{align*} %\numberthis
	\tau^{(K)}_{v.C\rightarrow u} & = \E (Y_{v.C}^{K-1} Y_u) \E (Y_{v.C}^2) - \E (Y_{v.C}^K) \E (Y_{v.C} Y_u)\\
	& = \E\left\{( \varepsilon_{Z_1} +  \pi_{vu.C}\varepsilon_{Z_2})^{K-1} Y_u \right\}
	\E\left\{(  \varepsilon_{Z_1} +  \pi_{vu.C}\varepsilon_{Z_2})^2\right\}\\
	&\qquad - \E\left\{(  \varepsilon_{Z_1} +  \pi_{vu.C}\varepsilon_{Z_2})^K\right\} \E\left\{(  \varepsilon_{Z_1} +  \pi_{vu.C}\varepsilon_{Z_2}) (Y_u)\right\}\\
	& =
	\textstyle\left\{\sum_{a = 0}^{K-1}{K-1 \choose a} \E\left(\varepsilon_{Z_1}^a\right) \E\left(\pi_{vu.C}^{K-1-a}\varepsilon_{Z_2}^{K-1-a}Y_u\right)\right\} \left\{\E(\varepsilon_{Z_1}^2)  + \E( \pi_{vu.C}^2\varepsilon_{Z_2}^2) \right\}
	\\
	&\textstyle\qquad - \left\{\sum_{a = 0}^K {K \choose a} \E\left(\varepsilon_{Z_1}^a \right) \E\left(\pi_{vu.C}^{K-a}\varepsilon_{Z_2}^{K - a}Y_u\right)\right\}
	\E\left( \pi_{vu.C}\varepsilon_{Z_2}Y_u\right) \\
	\qquad &\textstyle=\left\{\sum_{a = 1}^{K-1}{K-1 \choose a} \E\left(\varepsilon_{Z_1}^a\right)\pi_{vu.C}^{K-1 - a} \E\left(\varepsilon_{Z_2}^{K-a}\right)\right\} \left\{\E(\varepsilon_{Z_1}^2)  + \pi_{vu.C}^2\E(\varepsilon_{Z_2}^2) \right\}
	\\
	&\textstyle\qquad - \left\{\sum_{a = 1}^{K-1} {K \choose a} \E\left(\varepsilon_{Z_1}^a \right) \pi_{vu.C}^{K - a}\E\left(\varepsilon_{Z_2}^{K - a}\right)\right\}
	\pi_{vu.C}\E\left(\varepsilon_{Z_2}^2\right)
	\\
	&\textstyle\qquad + \pi_{vu.C}^{K-1}\E(\varepsilon_{Z_2}^K)\left\{\E(\varepsilon_{Z_1}^2) + \pi_{vu.C}^2\E(\varepsilon_{Z_2}^2)\right\} \\
	&\textstyle\qquad - \left\{\pi_{vu.C}^K\E(\varepsilon_{Z_2}^K) + \E(\varepsilon_{Z_1}^K) \right\}\pi_{vu.C}\E(\varepsilon_{Z_2}^2)\\
	&\textstyle = \left\{\sum_{a = 1}^{K-1}{K-1 \choose a} \pi_{vu.C}^{K-1 - a}\E\left(\varepsilon_{Z_1}^a\right) \E\left(\varepsilon_{Z_2}^{K-a}\right)\right\} \left\{\E(\varepsilon_{Z_1}^2)  + \pi_{vu.C}^2\E(\varepsilon_{Z_2}^2) \right\}
	\\
	&\textstyle\qquad - \left\{\sum_{a = 1}^{K-1} {K \choose a} \E\left(\varepsilon_{Z_1}^a \right) \pi_{vu.C}^{K - a}\E\left(\varepsilon_{Z_2}^{K - a}\right)\right\}
	\pi_{vu.C}\E\left(\varepsilon_{Z_2}^2\right) \\
	&\textstyle \qquad + \pi_{vu.C}^{K-1}\E(\varepsilon_{Z_2}^K)\E(\varepsilon_{Z_1}^2) - \pi_{vu.C}\E(\varepsilon_{Z_1}^K)\E(\varepsilon_{Z_2}^2).
	\end{align*}	
	When $K$ is odd, $\E(\varepsilon_{Z_1}^a)\E(\varepsilon_{Z_2}^{K-a}) = 0$ for all $a = 1, \ldots,K-1$, so we are left with
	\begin{align*}\numberthis \label{eq:emptyZ3Odd}
	\tau^{(K)}_{v.C\rightarrow u} &= \pi_{vu.C}^{K-1}\E(\varepsilon_{Z_2}^K)\sigma_{Z_1}^2 - \pi_{vu.C}\E(\varepsilon_{Z_1}^K)\sigma_{Z_2}^2\\
	& = \pi_{vu.C}^{K-1}\eta_{Z_2}\sigma_{Z_1}^2 - \pi_{vu.C}\eta_{Z_1}\sigma_{Z_2}^2.
	\end{align*}
	When $K$ is even, then $\E(\varepsilon_{Z_1}^a)\E(Y_u^{K-a}) = 0$ when $a$ is odd, so we are left with
	\begin{align*}%\numberthis
	\tau^{(K)}_{v.C\rightarrow u} &= \textstyle\left\{\sum_{a = 2,4,\ldots,K-2}{K-1 \choose a} \pi_{vu.C}^{K-1 - a}\E\left(\varepsilon_{Z_1}^a\right) \E\left(\varepsilon_{Z_2}^{K-a}\right)\right\}  \left\{\E(\varepsilon_{Z_1}^2)  + \pi_{vu.C}^2\E(\varepsilon_{Z_2}^2) \right\} \\
	&\textstyle\qquad - \left\{\sum_{a = 2,4,\ldots,K-2} {K \choose a} \E\left(\varepsilon_{Z_1}^a \right) \pi_{vu.C}^{K - a}\E\left(\varepsilon_{Z_2}^{K - a}\right)\right\}
	\pi_{vu.C}\E\left(\varepsilon_{Z_2}^2\right) \\
	& \textstyle\qquad + \pi_{vu.C}^{K-1}\E(\varepsilon_{Z_2}^K)\E(\varepsilon_{Z_1}^2) - \pi_{vu.C}\E(\varepsilon_{Z_1}^K)\E(\varepsilon_{Z_2}^2).
	\end{align*}
	Evaluating the moments yields	
	\begin{align*} %\numberthis
	\tau^{(K)}_{v.C\rightarrow u}&= \textstyle\left\{\sum_{a = 2,\ldots,K-2}{K-1 \choose a} \pi_{vu.C}^{K-1 - a}(a-1)!! \sigma_{Z_1}^a (K - a-1)!!\sigma_{Z_2}^{K-a}\right\}\left(\sigma_{Z_1}^2  + \pi_{vu.C}^2\sigma_{Z_2}^2 \right)
	\\
	&\textstyle\qquad - \left\{\sum_{a = 2,\ldots,K-2}{K \choose a} (a-1)!!\sigma_{Z_1}^a \pi_{vu.C}^{K - a}(K - a-1)!!\sigma_{Z_2}^{K-a}\right\}\left(\pi_{vu.C}\sigma_{Z_2}^2\right)\\
	&\textstyle \qquad + \pi_{vu.C}^{K-1}\E(\varepsilon_{Z_2}^K)\sigma_{Z_1}^2 - \pi_{vu.C}\E(\varepsilon_{Z_1}^K)\sigma_{Z_2}^2\\
	&\textstyle= \left\{\sum_{a = 2,\ldots,K-2}{K-1 \choose a} \pi_{vu.C}^{K-1 - a}(a-1)!! \sigma_{Z_1}^{a + 2} (K - a-1)!!\sigma_{Z_2}^{K-a}\right\}
	\\
	&\textstyle \qquad + \pi_{vu.C}\sigma_{Z_2}^2 \left\{\sum_{a = 2,\ldots,K-2}{K-1 \choose a} \pi_{vu.C}^{K - a}(a-1)!! \sigma_{Z_1}^a (K - a-1)!!\sigma_{Z_2}^{K-a}\right\} \\
	&\textstyle\qquad - \left\{\sum_{a = 2,\ldots,K-2}{K-1 \choose a}\frac{K}{K-a} (a-1)!!\sigma_{Z_1}^a \pi_{vu.C}^{K - a}(K - a-1)!!\sigma_{Z_2}^{K-a}\right\} \left(\pi_{vu.C}\sigma_{Z_2}^2\right)\\
	&\textstyle \qquad + \pi_{vu.C}^{K-1}\E(\varepsilon_{Z_2}^K)\sigma_{Z_1}^2 - \pi_{vu.C}\E(\varepsilon_{Z_1}^K)\sigma_{Z_2}^2\\
	&\textstyle= \left\{\sum_{a = 2,\ldots,K-2}{K-1 \choose a} \pi_{vu.C}^{K-1 - a}(a-1)!! \sigma_{Z_1}^{a + 2} (K - a-1)!!\sigma_{Z_2}^{K-a}\right\}
	\\
	&\textstyle\qquad + \left\{\sum_{a = 2,\ldots,K-2}{K-1 \choose a}\left(1 - \frac{K}{K-a}\right) (a-1)!!\sigma_{Z_1}^a \pi_{vu.C}^{K - a}(K - a-1)!!\sigma_{Z_2}^{K-a}\right\}\\
	&\textstyle\qquad \times \left(\pi_{vu.C}\sigma_{Z_2}^2\right)\\
	&\textstyle \qquad + \pi_{vu.C}^{K-1}\E(\varepsilon_{Z_2}^K)\sigma_{Z_1}^2 - \pi_{vu.C}\E(\varepsilon_{Z_1}^K)\sigma_{Z_2}^2\\
	&\textstyle= \left\{\sum_{a = 2,\ldots,K-4}{K-1 \choose a} \pi_{vu.C}^{K-1 - a}(a-1)!! \sigma_{Z_1}^{a + 2} (K - a-1)!!\sigma_{Z_2}^{K-a}\right\}
	\\
	&\textstyle\qquad - \left\{\sum_{a = 4,\ldots,K-2}{K-1 \choose a}\frac{a}{K-a} (a-1)!!\sigma_{Z_1}^a \pi_{vu.C}^{K - a} (K - a-1)!!\sigma_{Z_2}^{K-a}\right\} \left(\pi_{vu.C}\sigma_{Z_2}^2\right)\\
	&\textstyle \qquad + \pi_{vu.C}(K-1)!!\sigma_{Z_1}^K\sigma_{Z_2}^2 - \pi_{vu.C}^{K-1}(K-1)!!\sigma_{Z_2}^K \sigma_{Z_1}^2\\
	&\textstyle \qquad + \pi_{vu.C}^{K-1}\E(\varepsilon_{Z_2}^K)\sigma_{Z_1}^2 - \pi_{vu.C}\E(\varepsilon_{Z_1}^K)\sigma_{Z_2}^2.
	\end{align*}
	Rewriting terms and a change of variables show that the first
	two lines cancel leaving	$\tau^{(K)}_{v.C\rightarrow
		u}$ equal to
	
	\begin{align*}
	\numberthis   \label{eq:emptyZ3Even}
	& \sum_{a = 2,\ldots,K-4}\left\{ {K-1 \choose a+2}\frac{(a+1)(a+2)}{(K-(a + 1))(K-(a+2))} \pi_{vu.C}^{K- (a + 2)}\right. \\
	&\left. \quad \qquad \qquad \times (a-1)!! \sigma_{Z_1}^{a + 2} (K - a-1)!!\sigma_{Z_2}^{K-(a+2)}\vphantom{ {K-1 \choose a+2}} \right\}\pi_{vu.C}\sigma_{Z_2}^2
	\\
	&\quad - \left\{\sum_{a = 4,\ldots,K-2}{K-1 \choose a}\left(\frac{a}{K-a}\right) (a-1)!!\sigma_{Z_1}^a \pi_{vu.C}^{K - a}(K - a-1)!!\sigma_{Z_2}^{K-a}\right\}\pi_{vu.C}\sigma_{Z_2}^2\\
	& \quad + \pi_{vu.C}(K-1)!!\sigma_{Z_1}^K\sigma_{Z_2}^2 - \pi_{vu.C}^{K-1}(K-1)!!\sigma_{Z_2}^K \sigma_{Z_1}^2\\
	& \quad + \pi_{vu.C}^{K-1}\E(\varepsilon_{Z_2}^K)\sigma_{Z_1}^2 - \pi_{vu.C}\E(\varepsilon_{Z_1}^K)\sigma_{Z_2}^2\\
	&= \pi_{vu.C}\sigma_{Z_2}^2 \sum_{a = 2,\ldots,K-4} \left\{{K-1 \choose a+2}\frac{a+2}{K-(a+2)} \pi_{vu.C}^{K- (a + 2)}((a+ 2) - 1)!!\right. \\
	&\left. \hphantom{\pi_{vu.C}\sigma_{Z_2}^2 \sum_{a = 2,\ldots,K-4}} \qquad \times \sigma_{Z_1}^{a + 2} (K-(a+2) - 1)!!\sigma_{Z_2}^{K-(a+2)}\vphantom{ \pi_{vu.C}\sigma_{Z_2}^2 \sum_{a = 2,\ldots,K-4} }\right\}
	\\
	&\quad - \pi_{vu.C}\sigma_{Z_2}^2 \left\{\sum_{a = 4,\ldots,K-2}{K-1 \choose a}\left(\frac{a}{K-a}\right) (a-1)!!\sigma_{Z_1}^a \pi_{vu.C}^{K - a}(K - a-1)!!\sigma_{Z_2}^{K-a}\right\}\\
	& \quad + \pi_{vu.C}(K-1)!!\sigma_{Z_1}^K\sigma_{Z_2}^2 - \pi_{vu.C}^{K-1}(K-1)!!\sigma_{Z_2}^K \sigma_{Z_1}^2\\
	& \quad + \pi_{vu.C}^{K-1}\E(\varepsilon_{Z_2}^K)\sigma_{Z_1}^2 - \pi_{vu.C}\E(\varepsilon_{Z_1}^K)\sigma_{Z_2}^2\\
	&= \pi_{vu.C}(K-1)!!\sigma_{Z_1}^K\sigma_{Z_2}^2 - \pi_{vu.C}^{K-1}(K-1)!!\sigma_{Z_2}^K \sigma_{Z_1}^2  + \pi_{vu.C}^{K-1}\E(\varepsilon_{Z_2}^K)\sigma_{Z_1}^2 - \pi_{vu.C}\E(\varepsilon_{Z_1}^K)\sigma_{Z_2}^2\\
	&= \pi_{vu.C}^{K-1}\eta_{Z_2}\sigma_{Z_1}^2 - \pi_{vu.C}\eta_{Z_1}\sigma_{Z_2}^2.
	\end{align*}
	This is the same expression as when $K$ is odd. So for any $K > 2$,
	\begin{equation}
	\eta_{Z_1} \neq \frac{\pi_{vu.C}^{K-2}\sigma^2_{Z_1}\eta_2}{\sigma^2_{Z_2}} \qquad \iff \qquad \tau^{(K)}_{v.C\rightarrow u} \neq 0.
	\end{equation}
	Since $Z_1$ and $Z_2$ are disjoint, we can always select the $K$th moments of the individual error moments so that this holds.
	
	Now consider the case when $Z_3$ is not empty. From Equations~\eqref{eq:Total}, \eqref{eq:emptyZ3Odd}, and \eqref{eq:emptyZ3Even},
	\begin{align*}
	%\numberthis
	\tau_{v.C\rightarrow u} &= \pi_{vu.C}^{K-1}\eta_{Z_2}\sigma_{Z_1}^2 - \pi_{vu.C}\eta_{Z_1}\sigma_{Z_2}^2\\
	&\quad + \E\left\{( \varepsilon_{Z_1}+\pi_{vu.C}\varepsilon_{Z_2})^{K-1} Y_u \right\} \sigma_{Z_3}^2  \\
	&\quad + \E\left\{Y_u \sum_{\substack{|w| = K-1\\ w_3 > 0}}{K-1 \choose w}  \varepsilon_{Z_1}^{w_1} (\pi_{vu.C}  \varepsilon_{Z_2})^{w_2}  \varepsilon_{Z_3}^{w_3}\right\}
	\left(\sigma_{Z_1}^2 + \pi_{vu.C}\sigma_{Z_2}^2 + \sigma_{Z_3}^2\right)\\
	&\quad - \E\left\{( \varepsilon_{Z_1} + \pi_{vu.C}\varepsilon_{Z_2})^{K}\right\}\E\left( \varepsilon_{Z_3} Y_u\right)\\
	& \quad -  \E\left\{Y_u\sum_{\substack{|w| = K\\ w_3 >
			0}}{K \choose w}  \varepsilon_{Z_1}^{w_1} (\pi_{vu.C}
	\varepsilon_{Z_2})^{w_2}  \varepsilon_{Z_3}^{w_3}\right\}
	\left\{\pi_{vu.C}\sigma_{Z_2}^2 + \E\left( \varepsilon_{Z_3} Y_u\right)\right\}
	\\
	& =
	\pi_{vu.C}^{K-1}\eta_{Z_2}\sigma_{Z_1}^2 - \pi_{vu.C}\eta_{Z_1}\sigma_{Z_2}^2 \\
	&\quad + \left\{\pi_{vu.C}^{K-1} \left((K-1)!!\sigma_{Z_2}^K + \eta_{Z_2}\right)\right\} \sigma^2_{Z_3} + \sigma_{Z_3}^2\sum_{a = 1}^{K-1}\E(\varepsilon_{Z_1}^a)\pi_{vu.C}^{K-1-a}\E(\varepsilon_{Z_2}^{K-a})\\  &\quad + \E\left\{Y_u \sum_{\substack{|w| = K-1\\ w_3 > 0}}{K-1 \choose w}  \varepsilon_{Z_1}^{w_1} (\pi_{vu.C}  \varepsilon_{Z_2})^{w_2}  \varepsilon_{Z_3}^{w_3} \right\}
	\left(\sigma_{Z_1}^2 + \pi_{vu.C}\sigma_{Z_2}^2 + \sigma_{Z_3}^2\right)
	\\
	&\quad - \left[\vphantom{\sum_{a = 1}^{K-1}\E(\varepsilon_{Z_1}^a)} (K-1)!!\sigma_{Z_1}^K + \eta_{Z_1} + \pi_{vu.C}^{K}\left\{(K-1)!!\sigma_{Z_2}^K + \eta_{Z_2}\right\} \right. \\
	&\left.  \quad \quad + \sum_{a = 1}^{K-1}\E(\varepsilon_{Z_1}^a)\E(\varepsilon_{Z_2}^{K-a}\pi_{vu.C}^{K-a})\right]\E\left( \varepsilon_{Z_3} Y_u\right) \\
	&\quad -  \E\left\{Y_u \sum_{\substack{|w| = K\\ w_3 > 0}}{K \choose w}  \varepsilon_{Z_1}^{w_1} (\pi_{vu.C}  \varepsilon_{Z_2})^{w_2}  \varepsilon_{Z_3}^{w_3}\right\}
	\left\{\pi_{vu.C}\sigma_{Z_2}^2 + \E\left( \varepsilon_{Z_3} Y_u\right)\right\}.
	\end{align*}
	Since the unevaluated terms are fixed with respect to $\eta_{Z_1}$ and $\eta_{Z_2}$, selecting values such that
	\begin{align*}\numberthis \label{eq:eta_Confounding}
	\eta_{Z_1} &\neq \frac{1}{\pi_{vu.C}\sigma^2_{Z_2} +  \E(\varepsilon_{Z_3}Y_u)}\left[\vphantom{\E\left(Y_u \sum_{\substack{|w| = K\\ w_3 > 0}}{K \choose w}  \varepsilon_{Z_1}^{w_1} (\pi_{vu.C}  \varepsilon_{Z_2})^{w_2}  \varepsilon_{Z_3}^{w_3}\right)}
	\pi_{vu.C}^{K-1}\eta_{Z_2}\sigma_{Z_1}^2 \right.\\
	&\quad + \left\{\pi_{vu.C}^{K-1} (K-1)!!\sigma_{Z_2}^K + \eta_{Z_2}\right\} \sigma^2_{Z_3} + \sigma_{Z_3}^2\sum_{a = 1}^{K-1}\E(\varepsilon_{Z_1}^a)\pi_{vu.C}^{K-1-a}\E(\varepsilon_{Z_2}^{K-a})\\  &\quad + \E\left\{Y_u \sum_{\substack{|w| = K-1\\ w_3 > 0}}{K-1 \choose w}  \varepsilon_{Z_1}^{w_1} (\pi_{vu.C}  \varepsilon_{Z_2})^{w_2}  \varepsilon_{Z_3}^{w_3} \right\}
	\times
	\left(\sigma_{Z_1}^2 + \pi_{vu.C}\sigma_{Z_2}^2 + \sigma_{Z_3}^2\right)
	\\
	&\quad - \left\{(K-1)!!\sigma_{Z_1}^K + \pi_{vu.C}^{K}\left((K-1)!!\sigma_{Z_2}^K + \eta_{Z_2}\right) \vphantom{\sum_{a = 1}^{K-1}\E(\varepsilon_{Z_1}^a)\pi_{vu.C}^{K-a}}  \right.\\
	&\left. \qquad + \sum_{a = 1}^{K-1}\E(\varepsilon_{Z_1}^a)\pi_{vu.C}^{K-a}\E(\varepsilon_{Z_2}^{K-a})\right\}\E\left( \varepsilon_{Z_3} Y_u\right) \\
	&\left.\quad -  \E\left\{Y_u \sum_{\substack{|w| = K\\ w_3 > 0}}{K \choose w}  \varepsilon_{Z_1}^{w_1} (\pi_{vu.C}  \varepsilon_{Z_2})^{w_2}  \varepsilon_{Z_3}^{w_3}\right\}
	\left\{\pi_{vu.C}\sigma_{Z_2}^2 + \E\left( \varepsilon_{Z_3} Y_u\right)\right\} \right]
	\end{align*}
	implies that $\tau^{(K)}_{v.C \rightarrow u} \neq 0$.
	
\end{proof}
\begin{remark}\label{remark:confounding}
	Even if $u \not \in \pa(v)$ and $\pi_{vu.C} = 0$, as long as there is some $z \in Z_3$ such that $\pi_{vz.C} \neq 0$ and $\pi_{uz} \neq 0$ (i.e., $v$ and $u$ are confounded by $z$), then $\tau^{(K)}_{v.C \rightarrow u}$ is still non-zero for generic error moments. However, even if there is a $z$ with a directed path to both $u$ and $v$ which is not blocked by $C$, $\pi_{vz.C} \neq 0$ and $\pi_{uz} \neq 0$ is not necessarily implied by parental faithfulness because $z$ may be in $\an(v)\setminus \pa(v)$ or $\an(u)\setminus \pa(u)$.
\end{remark}
\begin{proof}
	When $\pi_{vu.C} = 0$, ~\eqref{eq:Total} reduces to
	\begin{align*}\numberthis\label{eq:justConfounding}
	\tau^{(K)}_{v.C\rightarrow u} & = \E (Y_{v.C}^{K-1} Y_u) \E (Y_{v.C}^2) - \E (Y_{v.C}^K) \E (Y_{v.C} Y_u)
	\\
	& = \E\left(\left(  \varepsilon_{Z_1} + \varepsilon_{Z_3} \right)^{K-1} Y_u \right)
	\E\left\{\left( \varepsilon_{Z_1} + \varepsilon_{Z_3}\right) ^2\right\}\\
	&\quad - \E\left\{\left(  \varepsilon_{Z_1} + \varepsilon_{Z_3} \right)^K\right\} \E\left\{\left(  \varepsilon_{Z_1} + \varepsilon_{Z_3} \right) Y_u\right\}\\
	& = \sum_{a = 0}^{K-1}\binom{K-1}{a}\E\left(\varepsilon_{Z_1}^a\varepsilon_{Z_3}^{K-1 - a} Y_u\right)\left(\sigma^2_{Z_1} + \sigma^2_{Z_3}\right)\\
	& \quad - \sum_{a = 0}^{K}\binom{K}{a}\E\left(\varepsilon_{Z_1}^a\varepsilon_{Z_3}^{K-a}\right)\E\left(\varepsilon_{Z_3} Y_u\right)\\
	& = \sum_{a = 0}^{K-1}\binom{K-1}{a}\E\left(\varepsilon_{Z_1}^a\varepsilon_{Z_3}^{K-1 - a} Y_u\right)\left(\sigma^2_{Z_1} + \sigma^2_{Z_3}\right)\\
	& \quad - \E(\varepsilon_{Z_1}^{K})\E(\varepsilon_{Z_3}Y_u) -  \sum_{a = 0}^{K-1}\binom{K}{a}\E\left(\varepsilon_{Z_1}^a\varepsilon_{Z_3}^{K-a}\right)\E\left(\varepsilon_{Z_3} Y_u\right).
	\end{align*}
	Since the first and third term do not involve $\E(\varepsilon_{Z_1}^{K})$, letting
	\begin{equation}
	\E\left(\varepsilon_{Z_1}^{K}\right) \neq \frac{\sum\limits_{a = 0}^{K-1}\binom{K-1}{a}\E\left(\varepsilon_{Z_1}^a\varepsilon_{Z_3}^{K-1 - a} Y_u\right)\left(\sigma^2_{Z_1} + \sigma^2_{Z_3}\right) -  \sum\limits_{a = 0}^{K-1}\binom{K}{a}\E\left(\varepsilon_{Z_1}^a\varepsilon_{Z_3}^{K-a}\right)\E\left(\varepsilon_{Z_3} Y_u\right)}{\E(\varepsilon_{Z_3}Y_u)}
	\end{equation} 
	ensures that $\tau^{(K)}_{v.C \rightarrow u} \neq 0$ so that the quantity is non-zero for generic error moments.
\end{proof}

\section{Proof of Lemma \ref{ch3:thm:errorBeta}}
\begin{proof}
	We use $\|\cdot\|$ to denote vector norms and
	$\vertiii{\cdot}$ to denote matrix norms. Conditions
	\ref{con:min_eigen} and \ref{con:bounded_moments} imply that
	\begin{equation}\label{eq:betaBound}
	\begin{aligned}
	||\beta_{vC}||_\infty &\leq ||\beta_{vC}||_2 \leq \vertiii{\Sigma_{CC}^{-1}}_2 \|\Sigma_{Cv}\|_2 \leq \vertiii{\Sigma_{CC}^{-1}}_2 \sqrt{J}\|\Sigma_{Cv}\|_\infty \leq  \frac{\sqrt{J} M}{\lMin}.
	\end{aligned}
	\end{equation}		
	Condition \ref{con:sample_moments} implies $\hat \Sigma_{CC} = \Sigma_{CC} + E$ and $\hat \Sigma_{Cv} = \Sigma_{Cv} + e$ with $\|E\|_\infty < \delta_1$ and $\|e\|_\infty < \delta_1$. Using results from \citet[Equation 5.8.7]{horn2013matrix} for the third and fourth inequalities below yields
	\begin{align*}
	\|\hat \beta_{vC} - \beta_{vC}\|_\infty & \leq \|\hat \beta_{vC} - \beta_{vC}\|_2 = \|(\Sigma_{CC} + E)^{-1}(\Sigma_{Cv} + e) - \Sigma_{CC}^{-1}\Sigma_{Cv}\|_2 \\
	&\leq \|(\Sigma_{CC} + E)^{-1}\Sigma_{Cv} - \Sigma_{CC}^{-1}\Sigma_{Cv} \|_2 + \|(\Sigma_{CC} + E)^{-1} e\|_2\\
	& \leq \frac{\vertiii{\Sigma_{CC}^{-1}E}_2}{1 - \vertiii{\Sigma_{CC}^{-1}E}_2}\|\beta_{vC}\|_2 +  \vertiii{(\Sigma_{CC} + E)^{-1}}_2 \|e\|_2\\ \numberthis \label{eq:boundedBeta}
	& \leq \frac{\vertiii{\Sigma_{CC}^{-1}E}_2}{1 - \vertiii{\Sigma_{CC}^{-1}E}_2}\|\beta_{vC}\|_2 +  \frac{1/\lMin}{1 - \vertiii{\Sigma_{CC}^{-1}E}_2}\|e\|_2\\
	&\leq \frac{\vertiii{\Sigma_{CC}^{-1}E}_2}{1 - \vertiii{\Sigma_{CC}^{-1}E}_2}\|\beta_{vC}\|_2 +  \frac{\sqrt{J}\delta_1/\lMin}{1 - \vertiii{\Sigma_{CC}^{-1}E}_2}.
	\end{align*}
	The term $ \vertiii{\Sigma_{CC}^{-1}E}_2 \leq \vertiii{\Sigma_{CC}^{-1}}_2 \vertiii{E}_2 \leq \frac{J\|E\|_\infty}{\lMin} \leq \frac{J\delta_1}{\lMin} < 1/2$. Since the bound in \eqref{eq:boundedBeta} is increasing in each of its arguments for $\vertiii{\Sigma_{CC}^{-1}E} < 1$,
	\begin{multline}
	\|\hat \beta_{vC} - \beta_{vC}\|_\infty  \leq \frac{\frac{J\delta_1}{\lMin}}{1 - \frac{J\delta_1}{\lMin}}\frac{\sqrt{J} M}{\lMin}+  \frac{\sqrt{J}\delta_1 / \lMin}{1 - \frac{J\delta_1}{\lMin} }
	= \frac{\sqrt{J}\delta_1/\lMin}{1 - \frac{J\delta_1}{\lMin} }\left(JM/ \lMin +1\right)\\
	\leq 2\frac{\sqrt{J}\delta_1}{\lMin}\left(JM/ \lMin +1\right) 
	\leq 4 \frac{J^{3/2}M\delta_1}{\lMin^2}= \delta_2.
	\end{multline}
	The penultimate inequality holds because by assumption $\frac{J\delta_1}{\lMin} < 1/2$ and the last inequality holds because by assumption $ M > \frac{\lMin}{J}$. 
\end{proof}

\section{Proof of Lemma~\ref{ch3:thm:errorMoment}}
\begin{proof}
	For $a \in \mathbb{Z}^{|C|+1}_{\geq 0}$, let $\binom{s}{a} = \frac{s}{a_1! a_2! \ldots a_{|C|+1}!}$ be the multinomial coefficient. Define the map %$f$ as
	\begin{equation*}
	\begin{aligned}
	f\left(\beta_{vC}, \left\{m_{V, \alpha}\right\}_{|\alpha| = s+r} \right) &= \E_P\left(Z_{v.C}^s Z_{u}^r\right) = \E_P\left\{\left(Z_{v} - \sum_{c \in C} \beta_{vc.C}Z_{c}\right)^s Z_{u}^r\right\}
	\\
	&= \E_P\left\{Z_{u}^r\sum_{|a| = s} \binom{s}{a} \prod_{c \in C}(- \beta_{vc.C}Z_{c})^{a_c} Z_{v}^{a_v}\right\}\\
	&= \sum_{|a| = s}\left\{ \binom{s}{a} m_{\left(C, v, u\right),\left(a, r\right)} \prod_{c \in C}\left(-\beta_{vc.C} \right)^{a_c} \right\}.
	\\
	\end{aligned}
	\end{equation*}
	Since $a$ is of length $|C| + 1$, there are $\binom{|C| + 1 + s - 1}{|C| + 1 - 1}$ moments we consider. By Condition~\ref{con:bounded_moments} and \ref{con:sample_moments}, each of the sample moments of $Y$ is restricted to $\left(-M - \delta_1, M + \delta_1\right)$, and as shown in \eqref{eq:betaBound} each of the $\hat \beta_{vz.C}$  is restricted to $\left(-\frac{J^{1/2}M}{\lMin} - \delta_2, \frac{J^{1/2}M}{\lMin} + \delta_2\right)$. In this domain, the partial derivatives of $f$ are bounded with
	\begin{align*}
	\left|\frac{\partial f}{m_{V, \alpha} } \right| &\leq s! \left(\frac{\sqrt{J}M}{\lMin} + \delta_2\right)^{s}, \text{ and }\\
	\left|\frac{\partial f}{\partial \beta_{vz.C}} \right|  & \leq \sum_{\substack{|a| = s\\ a_z > 0}} \binom{s}{a} (a_z) \left|  m_{(C,v,u),(a, r)} (-\beta_{vz.C})^{a_z - 1} \prod_{c \in C\setminus z}(-\beta_{vc.C})^{a_c} \right|
	\\
	& \leq \sum_{\substack{|a| = s\\ a_z > 0}} \binom{s}{a}s\left\{\left(M + \delta_1\right) \left(\frac{\sqrt{J}M}{\lMin} + \delta_2\right)^{s - 1}\right\}
	\\
	& \leq (|C|+1)^{s} s\left\{\left(M + \delta_1\right) \left(\frac{\sqrt{J}M}{\lMin} + \delta_2\right)^{s - 1} \right\}.
	\end{align*}
	By the mean value theorem for some $\left(\tilde \beta_{vC}, \left\{\tilde m_{V,\alpha}\right\}_{|\alpha| = s+r}\right)$, a convex combination of $\left(\hat \beta_{vC}, \left\{\hat m_{V, \alpha}\right\}_{|\alpha| = s+r}\right)$ and $\left(\beta_{vC}, \left\{ m_{V, \alpha}\right\}_{|\alpha| = s+r}\right)$,
	\begin{align*}
	\left|f \vphantom{\left(\beta_{vC}, \left\{m_{V, \alpha}\right\}_{|\alpha| = s+r} \right) - f\left(\hat \beta_{vC}, \left\{\hat m_{V, \alpha}\right\}_{|\alpha| = s+r} \right)}\right.&\left.\left(\beta_{vC}, \left\{m_{V, \alpha}\right\}_{|\alpha| = s+r} \right) - f\left(\hat \beta_{vC}, \left\{\hat m_{V, \alpha}\right\}_{|\alpha| = s+r} \right) \right| 
	\\
	&= \left|\left\{\nabla f\left(\tilde \beta_{vC}, \left\{\tilde m_{V, \alpha}\right\}_{|\alpha| = s+r}\right)\right\}^T \left\{ \left(\beta_{vC}, \left\{m_{V, \alpha}\right\}_{|\alpha| = s+r}\right) -   \left(\hat \beta_{vC}, \left\{\hat m_{V, \alpha}\right\}_{|\alpha| = s+r}\right) \right\} \right| 
	\\
	&\leq \left|\left\{\nabla_\beta f \left(\tilde \beta_{vC}, \left\{\tilde m_{V,\alpha}\right\}_{|\alpha| = s+r}\right)\right\}^T \left\{ \beta_{vC} - \hat \beta_{vC}\right\} \right|
	\\
	& \quad + \left|\left\{\nabla_m f \left(\tilde \beta_{vC}, \left\{\tilde m_{V, \alpha}\right\}_{|\alpha| = s+r}\right)\right\}^T \left\{ \left\{m_{V, \alpha}\right\}_{|\alpha| = s+r}  - \left\{\hat m_{V,\alpha}\right\}_{|\alpha| = s+r} \right\} \right| \\
	&\leq |C|\delta_2 \max\left|\frac{\partial f}{\partial \beta_{vz.C}} \right| + \binom{|C| + s}{|C|} \delta_1 \max \left|\frac{\partial f}{m_{V,\alpha}} \right|
	\end{align*}
	where $\nabla_\beta$ and $\nabla_m$ indicate the gradient with respect to the linear coefficients and moments, respectively. The last inequality follows from H\"{o}lder's inequality. Plugging in $\delta_2$ from Lemma \ref{ch3:thm:errorBeta} yields
	\begin{align*}
	\binom{|C| + s}{|C|} \delta_1 & \max \left|\frac{\partial f}{m_V, \alpha } \right| +  |C|\delta_2 \max\left|\frac{\partial f}{\partial \beta_{vz.C}} \right|
	\\
	& \leq \binom{|C| + s}{|C|} \delta_1 s! \left(\frac{\sqrt{J}M}{\lMin} + 4 \frac{J^{3/2}M\delta_1}{\lMin^2}\right)^{s}
	\\
	&\quad  + |C|4 \frac{J^{3/2}M\delta_1}{\lMin^2} (|C|+1)^{s} \left\{\left(M + \delta_1\right) s \left(\frac{\sqrt{J}M}{\lMin} + 4 \frac{J^{3/2}M\delta_1}{\lMin^2}\right)^{s - 1} \right\}
	\\
	& \leq (|C| + s)^s \delta_1 \left(\frac{3\sqrt{J}M}{\lMin}\right)^{s}
	\\
	&\quad  + |C|4 \frac{J^{3/2}M\delta_1}{\lMin^2} (|C|+1)^{s} \left\{\left(M + \delta_1\right) s \left(\frac{3\sqrt{J}M}{\lMin}\right)^{s - 1} \right\}
	\\
	& \leq (|C| + s)^s \delta_1 \left(\frac{3\sqrt{J}M}{\lMin}\right)^{s}
	\\
	&\quad  + |C|4 \frac{J\delta_1}{\lMin} (|C|+1)^{s} \left[2M s \left(\frac{3\sqrt{J}M}{\lMin}\right)^{s} \right]
	\\
	& \leq \delta_1 \left\{8(J+ K)^{K} JK \left(\frac{3\sqrt{J}M}{\lMin}\right)^{K}\left(1 + \frac{JM}{\lMin}\right)\right\}
	\\
	& \leq \delta_1 \left\{16(J+ K)^{K}  JK \left(\frac{3\sqrt{J}M}{\lMin}\right)^{K}\left(\frac{JM}{\lMin}\right)\right\}
	\\
	& = \delta_1 \left\{16(3^K)(J+ K)^{K}  K
	\frac{J^{(K+4)/2}M^{K+1}}{\lMin^{K+1}}\right\}. 
	\end{align*}
	The second inequality holds because we assumed $J\delta_1 /\lMin < 1 /2$; the third inequality holds because we assumed $\delta_1 < M$; the fourth inequality holds because we assumed $|C| \leq J$ and $s \leq K$; the fifth inequality holds because we assumed $JM / \lMin > 1$. 
\end{proof}

\section{Proof of Lemma~\ref{ch3:thm:errorTau}}

\begin{proof}
	
	Similar to the previous notation where $m_{H,\alpha} = \E(\prod_{h \in H} Z_{h}^{\alpha_h})$, we also allow for $v.C \in H$ indicating the population moment involving $Z_{v.C}$. By the triangle inequality,
	\begin{equation}
	\begin{aligned}
	|\hat \tau_{v.C \rightarrow u} -  \tau_{v \rightarrow u} |& = \left|\hat m_{(v.C,u),(K-1,1)} \hat m_{(v.C),(2)} - \hat m_{(v.C),(K)} \hat m_{(v.C,u),(1,1)} \right.\\
	&\quad \left.
	- (m_{(v.C,u),(K-1,1)} m_{v.C(2)} - m_{(v.C),(K)} m_{(v.C,u),(1,1)})\right|\\
	& \leq |\hat m_{(v.C,u)(K-1,1)} \hat m_{(v.C), (2)} - m_{(v.C,u),(K-1,1)} m_{(v.C),(2)}|\\
	&\quad  + | \hat m_{(v.C),(K)} \hat m_{(v.C),u),(1,1)}
	-  m_{(v),(K)} m_{(v,u),(1,1)})|.
	\\
	\end{aligned}
	\end{equation}
	
	Consider each of the two terms separately. For some $0 < \eta_1 < \delta_1 \Phi(J, K, M, \lMin)$  and $0 < \eta_2 < \delta_1 \Phi(J, K, M, \lMin)$ we have
	\begin{equation*}
	\begin{aligned}
	|\hat m_{(v.C,u),(K-1,1)} &\hat m_{(v.C),(2)} -
	m_{(v.C,u),(K-1,1)} m_{(v.C),(2)}| \\
	&=\left|(m_{(v.C,u),(K-1,1)} + \eta_1) (m_{(v.C),(2)} + \eta_2) - m_{(v.C,u),(K-1,1)} m_{(v.C),(2)}\right|
	\\
	&= |(m_{(v.C,u),(K-1,1)}\eta_2 + m_{(v.C),(2)}\eta_1) + \eta_1\eta_2|
	\\
	&\leq M \eta_2 + M \eta_1 + \eta_1 \eta_2
	\\
	& = 2M \delta_1 \Phi(J, K, M, \lMin) + \left(\delta_1 \Phi(J, K, M, \lMin)\right)^2.
	\end{aligned}
	\end{equation*}
	Using the analogous argument for the second term, we can bound the entire quantity as
	\[|\hat \tau_{v.C \rightarrow u} -  \tau_{v.C \rightarrow u} |
	< \delta_3 =  4M \delta_1 \Phi(J, K, M, \lMin) +
	2\left(\delta_1 \Phi(J, K, M, \lMin)\right)^2.
	\]
\end{proof}

\section{Pruning Procedure}
At the beginning of each step $z$ in Algorithm~\ref{alg:topOrder}, we have a set of nodes which have already been ordered, $\Theta^{(z-1)}$, and a set of nodes which have not yet been ordered, $\Psi^{(z-1)}$. To select the next node in the ordering, we calculate for each $v \in \Psi^{(z-1)}$,
\begin{equation*}
\hat T(v, \mathcal{C}_v^{(z)}, \Psi^{(z-1)}\setminus v),
\end{equation*}
where $\mathcal{C}_v^{(z)}$ is a set of potential parents we consider for node $v$. This requires fitting $\binom{|\mathcal{C}_v^{(z)}|}{J}$ regressions which can be computationally prohibitive when $|\mathcal{C}_v^{(z)}|$ is large. In order to speed up computation, rather than letting $\mathcal{C}_v^{(z)} = \Theta^{(z-1)}$, we seek to keep $\mathcal{C}_v^{(z)}$ small while preserving the theoretical guarantees in Theorem~\ref{ch3:thm:deterministicCorrect} and Corollary~\ref{ch3:thm:probGuarantee}. In particular, we let
\begin{equation}
\mathcal{C}^{(z)}_v = \Bigg\{p \in \mathcal{C}^{(z-1)}_v: \min_{C \in D^{(z)}_v} |\hat \tau_{v.C \rightarrow p}| > g^{(z)}\Bigg\} \;\cup\; \Theta^{(z-1)}_{z-1}
\end{equation}
where $D^{(z)}_v = \bigcup_{d < z} \{C: C \subseteq \mathcal{C}_v^{(d)} \setminus \{p\}; |C| \leq J \}$ and $g^{(z)}$ is some cut-off value. Suppose $g^{(z)} = \max(g^{(z-1)}, \alpha \hat T(r, \mathcal{C}^{(z)}_r, \Psi^{(z-1)}) )$ where $r$ is the root selected at step $z - 1$ and $\alpha$ is some tuning parameter in $[0,1]$. Under the assumptions of Theorem~\ref{ch3:thm:deterministicCorrect} where $\gamma$ is the signal strength, $g^{(z)} < \gamma /2$ for all steps $z$ because $\hat T(r, \mathcal{C}^{(z)}_r, \Psi^{(z-1)})) < \gamma / 2$. Thus, no parent of $v$ is mistakenly excluded from $\mathcal{C}^{(z)}_v$. 

This pruning procedure can lead to large computational savings; however, the savings, even with an oracle pruning parameter $g^{(z)} = \gamma/2$, will depend on the structure of the true graph. If there is not a unique total ordering, the savings may vary from one sample to another, even if $\hat G$ is the same, because the pruning depend on the topological ordering selected. On one extreme, if the true graph is the empty graph, then $|\mathcal{C}_v^{(z)}| = 1$ (specifically $\mathcal{C}_v^{(z)} = \Theta^{(z-1)}_{z-1}$) for all $v$ at all steps $z$. If the graph is a single chain, i.e., $1 \rightarrow 2 \rightarrow \ldots \rightarrow p$, then $|\mathcal{C}_v^{(z)}| = J+1$ where $J$ is the user specified maximum in-degree. However, there are also cases where the maximum in-degree is bounded, but $\max_{v \in V, z \in [p]}|\mathcal{C}_v^{(z)}| = O(p)$. For the graph shown in Figure~\ref{fig:upsideTree}, in the worst case, at step $z = 5$, if $\Theta^{(4)} = \{1,2,3,4\}$, then $\mathcal{C}_7^{(5)} = \{1,2,3,4\}$. However, if $\Theta^{(4)} = \{1,2,5, 3\}$, then $\mathcal{C}_7^{(5)} = \{5,3\}$. In general, for ``upside-down'' binary trees, in the worst case scenario, $\max_{v \in V, z \in [p]}|\mathcal{C}_v^{(z)}| \geq p/2$.

\begin{figure}[htb]
	\centering
	\begin{tikzpicture}[->,>=triangle 45,shorten >=1pt,
	auto,
	main node/.style={ellipse,inner sep=0pt,fill=gray!20,draw,font=\sffamily,
		minimum width = .5cm, minimum height = .5cm}]
	
	\node[main node] (7) {7};
	\node[main node] (6) [above right= .8cm of 7]  {6};
	\node[main node] (5) [above left= .8cm of 7]  {5};
	\node[main node] (4) [above right= .3cm of 6]  {4};
	\node[main node] (3) [above left= .3cm of 6]  {3};
	\node[main node] (2) [above right= .3cm of 5]  {2};
	\node[main node] (1) [above left= .3cm of 5]  {1};

	\path[color=black!20!blue,style={->}]
	(1) edge (5)
	(2) edge (5)
	(3) edge (6)
	(4) edge (6)
	(5) edge (7)
	(6) edge (7)
	;
	
	\end{tikzpicture}
	\caption{\label{fig:upsideTree}In the worst case, the pruning procedure does not lead to substantial computational savings for an ``upside down'' tree.}
\end{figure}
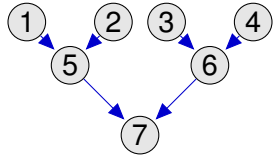

To reduce redundant computation required for the pruning procedure, we keep a running record of the minimum $|\tau_{v.C \rightarrow p}|$ for each $v, p \in V$. Thus, at each new step, we only need to consider the new sets in $ D^{(z-1)}_v \setminus  D^{(z)}_v$. Since $\mathcal{C}_{v}^{(z)}$ grows by at most 1 node at each step, then $|D^{(z-1)}_v \setminus  D^{(z)}_v| \leq \binom{\mathcal{C}^{(z)}_v - 1}{J-1}$. Also, when calculating the pruning statistics, we avoid re-computing $Y_{v.C}$, which is the most expensive part, since it was already computed when calculating $|\tau_{v.C \rightarrow u}|$ for some $u \in \Psi^{(z-2)}$ at the previous step.

\section{Additional Simulations}
\subsection{Bivariate Data}
In Figure~\ref{fig:compareBivariate} we compare the proportion of times DirectLiNGAM, Pairwise LiNGAM, and the proposed procedure are able to detect the correct causal direction in bivariate data. This allows for a direct comparison of the test statistics without the additional algorithmic changes. For each simulation, we let $\varepsilon_X$ and $\varepsilon_Y$ be either gamma or uniform with standard deviations randomly selected in $(.8,1)$. We then let $X = \varepsilon_X$ and $Y = \beta_{YX} X + \varepsilon_Y$ with $\beta_{YX}$ drawn randomly from $\pm(.65, 1)$. For $\tau^{(K)}$, we consider $K = 3,4$ and also consider the sum of statistics of the form $|\tau^{(K_1)}| + |\tau^{(K_2)}|$ which can in some cases have more accuracy. For the uniform distribution we let $(K_1, K_2) = (3,4), (4,6)$ and for the gamma distribution we let $(K_1, K_2) = (3,4), (3,5)$. For DirectLiNGAM and Pairwise LiNGAM we use the default settings in the code provided on the author's website\footnote{https://sites.google.com/site/sshimizu06/Dlingamcode}. We use $10,000$ simulations at each level of $n$. We see that at the smaller sample sizes, the proposed statistics have higher accuracy, especially the sum statistics. However, as $n$ grows larger, the DirectLiNGAM and Pairwise LiNGAM methods perform better. Since the third moments of the uniform are the same as the Gaussian, $\tau^{(3)}$ is a very poor determinant of causal direction when the errors are uniform. It still does slightly better than chance because $\tau^{(3)}_{Y\rightarrow X}$ tends to have a higher variance than $\tau^{(3)}_{X\rightarrow Y}$, but the accuracy actually decreases as $n$ increases.  

\begin{figure}[h]\centering
	\includegraphics[scale=.4]{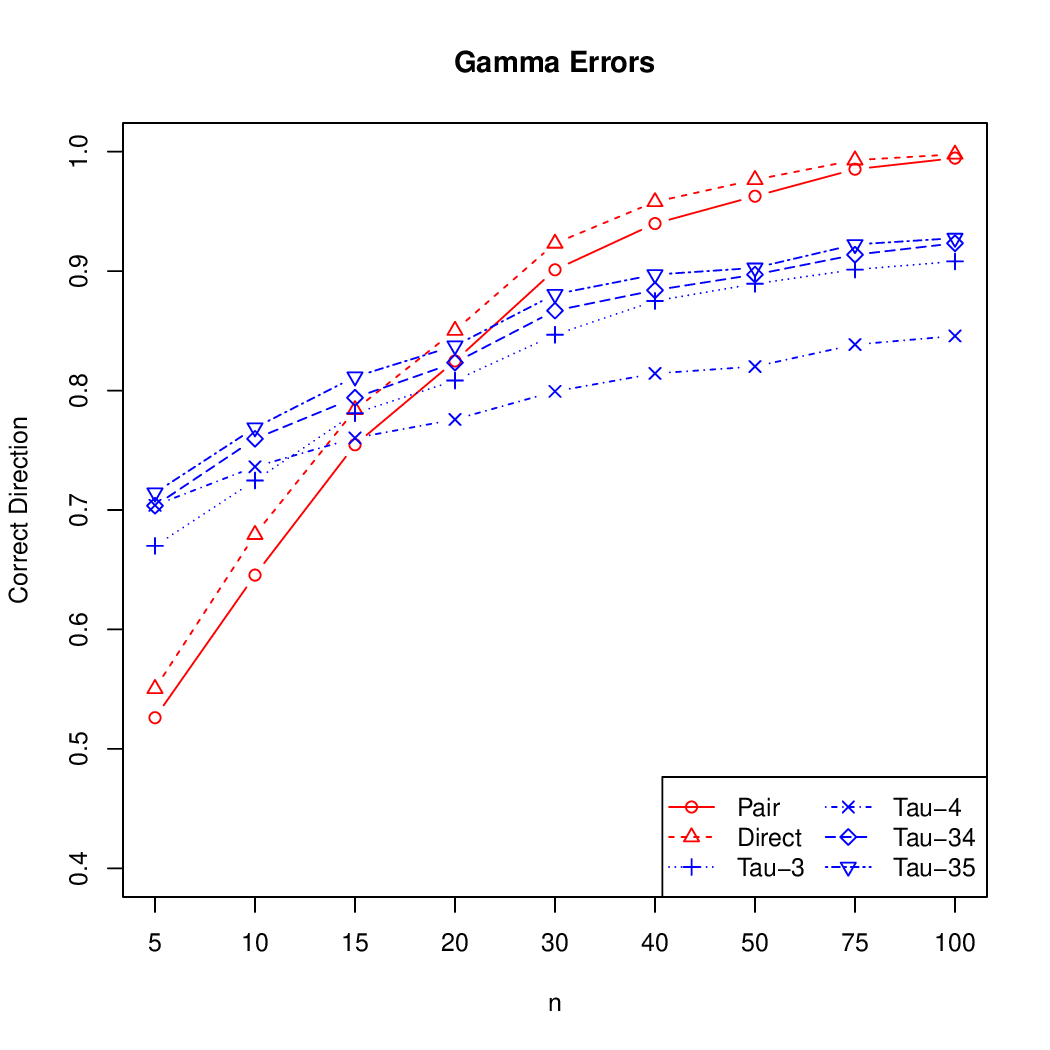}
	\includegraphics[scale=.4]{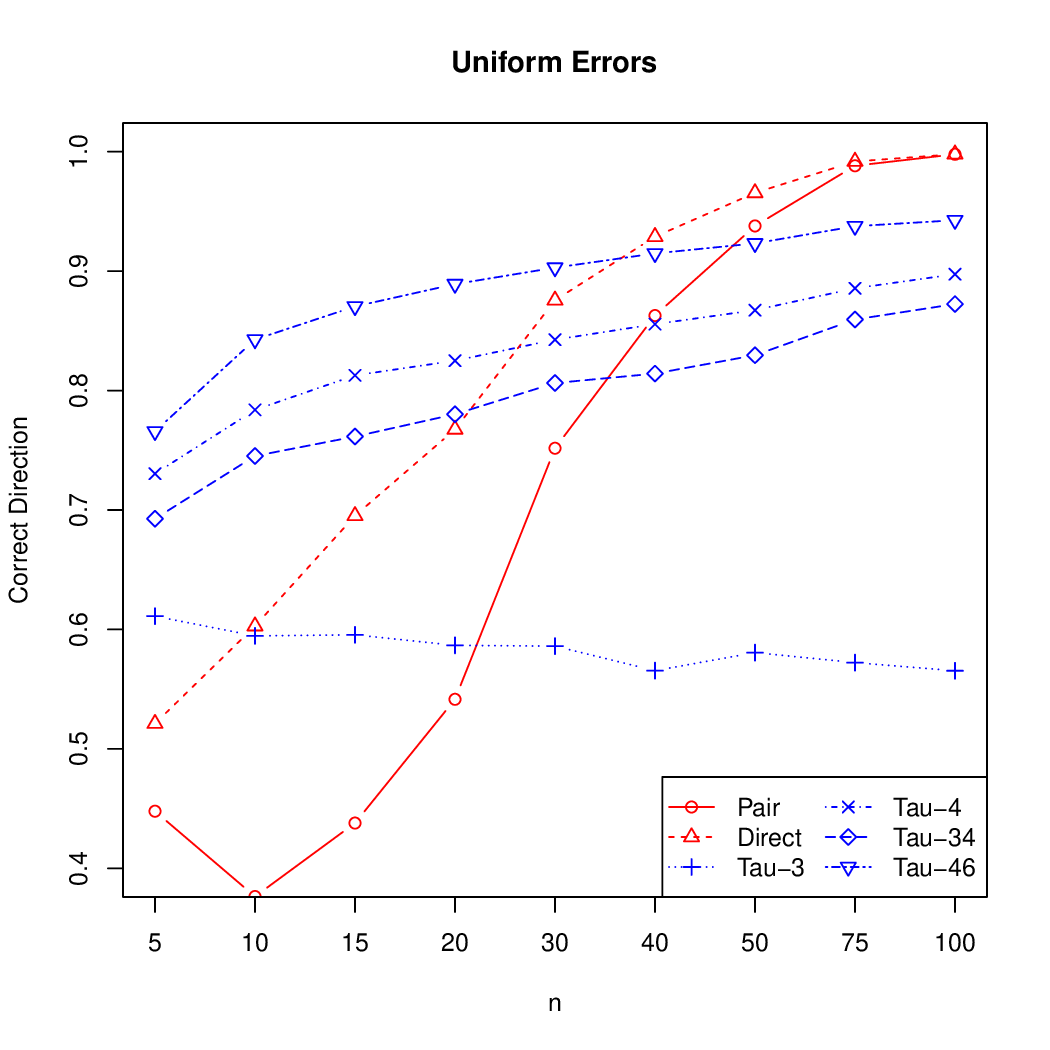}
	\caption{\label{fig:compareBivariate}Comparison of methods on bivariate data.}
\end{figure}

\subsection{Gamma and Gaussian errors}
We also show two additional simulation settings with the same setup as the high-dimensional non-hub setting from Section~\ref{sec:highDSims} where $n = 3/4p$. The results in the main text use uniform random variables, and here we give results for when the random variables are gamma and Gaussian. In particular, in the gamma setting, we let $\varepsilon_{vi} = \sigma_v(\alpha_{vi} - \frac{2}{\sqrt{2}})$ where $\alpha_{vi}$ is a gamma random variable with shape parameter $2$ and rate parameter $\sqrt{2}$ so that $\varepsilon_{vi}$ has mean 0 and variance $\sigma_v^2$. In the Gaussian case, we let $\varepsilon_{vi} \sim N(0, \sigma_v^2)$. In both cases, we draw $\sigma_v$ uniformly from $[.8, 1]$.

For the gamma case, as the theory would predict, the performance increases as the number of nodes and sample size increase. However, because the sample moments of gamma random variables concentrate slower than the uniform random variables, at each setting of $p$ and $n$, the performance with gamma errors is worse than the uniform case. 

For the Gaussian case, as the theory would predict, we see that the
estimated ordering does not improve with increasing $p$ and
$n$. However, is still better than chance. We posit this is because
even though the test statistic cannot determine causal direction with
Gaussian errors, it can still detect uncontrolled confounding (see
supplement Remark~\ref{remark:confounding}). Consider, step $z$ of the
procedure, where $\Theta^{(z)}$ is the ``already ordered" set and
$\Psi^{(z)}$ is the ``yet to be ordered" set. When selecting a root
from the subgraph induced by $\Psi^{(z)}$, node $v$ will be selected
as a root (using population quantities) only if the confounders,
$\pa(v) \cap \pa(\Psi^{(z)}\setminus v)$, have already been selected
into $\Theta^{(z)}$ and can be used to condition $v$. For example if
nodes $v,u \in \Psi^{(z)}$, share a common parent $s \in \Psi^{(z)}$,
$T(v, \Theta^{(z)}, \Psi^{(z)})$ and $T(u, \Theta^{(z)}, \Psi^{(z)})$
will still be positive because $|\tau_{v.C \rightarrow u}| > 0$ when
$s \not \in C$.  Thus, similar to how the PC algorithm can use v-structures
to orient some edges, the proposed procedure can use confounding
structures to identify some, but not all, causal orderings even in the Gaussian case.

\begin{figure}[h]\centering
	\includegraphics[scale=.5]{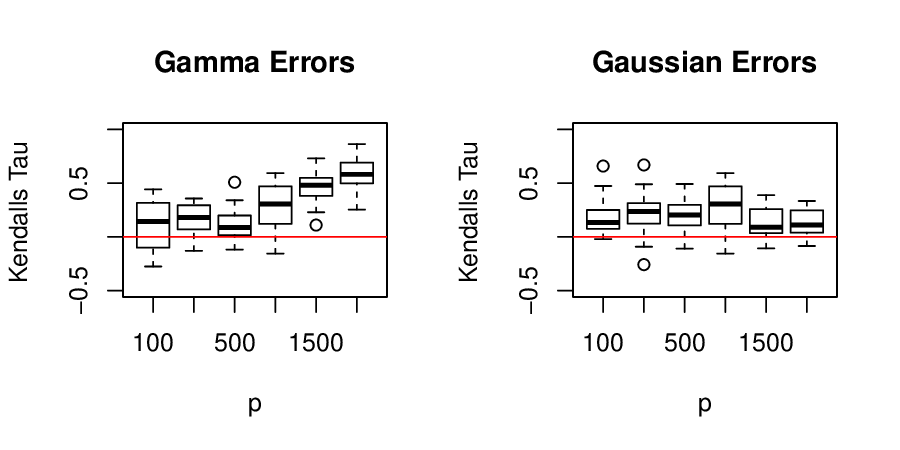}
	\caption{\label{fig:compareOtherErrors}Results of graph
		estimation when the errors are gamma or Gaussian.  Each
		sample size is set to $n = 3/4p$ where $p$ is the number of variables. Each barplot represents 20 simulation runs.}
\end{figure}	

\subsection{Pre-selection with hub graphs}	
Figure~\ref{fig:preSelectHub} shows the results of using pre-selection with randomly generated hub graphs. We use the same data generating procedure as described in Section~\ref{sec:highDSims}. Again we see that the proposed method, with or without pre-selection, outperforms the sparse Pairwise LiNGAM method. 
\begin{figure}[h]\centering
	\includegraphics[scale=.6]{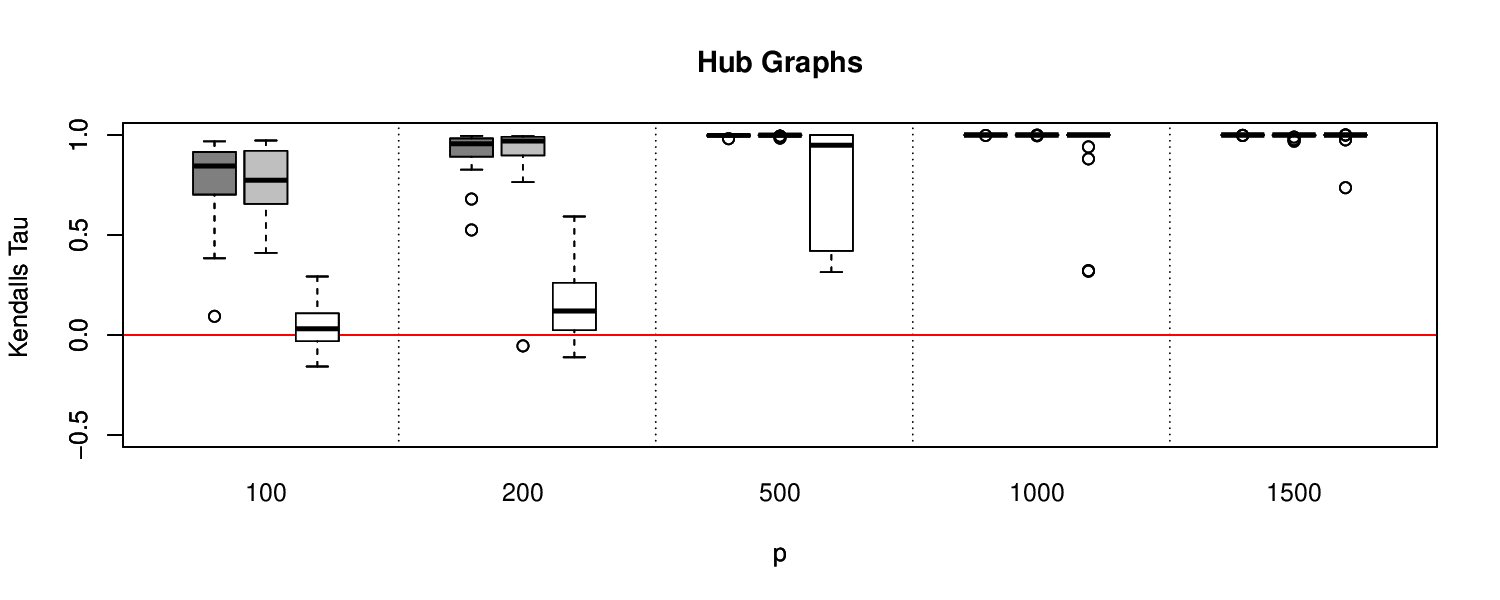}
	\caption{\label{fig:preSelectHub}Each boxplot represents the results of 20 simulations with random hub graphs and a pre-selection step; in each case $n = 3/4p$. From left to right the methods are: the proposed high dimensional LiNGAM procedure, same as Figure~\ref{fig:highDCons}; the proposed high dimensional LiNGAM procedure with pre-selection; the two stage pairwise procedure from \citet[Section 3.3]{hyvarinen2013pairwise}.}
\end{figure}	

\subsection{Parent selection using best subset regression}
At each step, as an alternative to picking a root by minimizing $|\tau|$ across all subsets $C \subseteq \Theta^{(z)}$ such that $|C| = J$, we could use a two-step procedure. Specifically, for each $v \in \Psi^{(z)}$, we first select a set of possible parents via best subset regression so that $C_v^{(z)}$ minimizes the conditional variance
\[C_v^{(z)} = \min_{C: C \subseteq \Theta^{(z)}; |C| = J} \hat \sigma^2_v - \hat \Sigma_{vC}\hat \Sigma_{CC}^{-1}\hat \Sigma_{Cv}. \]   	
We then only calculate $|\tau|$ for that specific set of parents and select a root $r$ by
\begin{equation}
r = \arg \min_{v \in \Psi^{(z)}} \max |\hat \tau_{v.C_v^{(z)} \rightarrow u}|.
\end{equation}	
If the parents of the roots are consistently selected by best subset regression, this procedure may also yield consistent estimation of the causal graph. This procedure works well in simulations; however, theoretical guarantees would require slightly different assumptions, such as a beta-min condition to ensure consistent parent selection. Also, in practice, we use a branch and bound procedure~\citep{leaps2017} and implementing a pruning procedure is not as straightforward so the two-step method tends to be more computationally expensive. The simulation settings use the random graphs as described in Section 4.2 with $n = 3/4 p$. 

\begin{figure}[h]\centering
	\includegraphics[scale=.5]{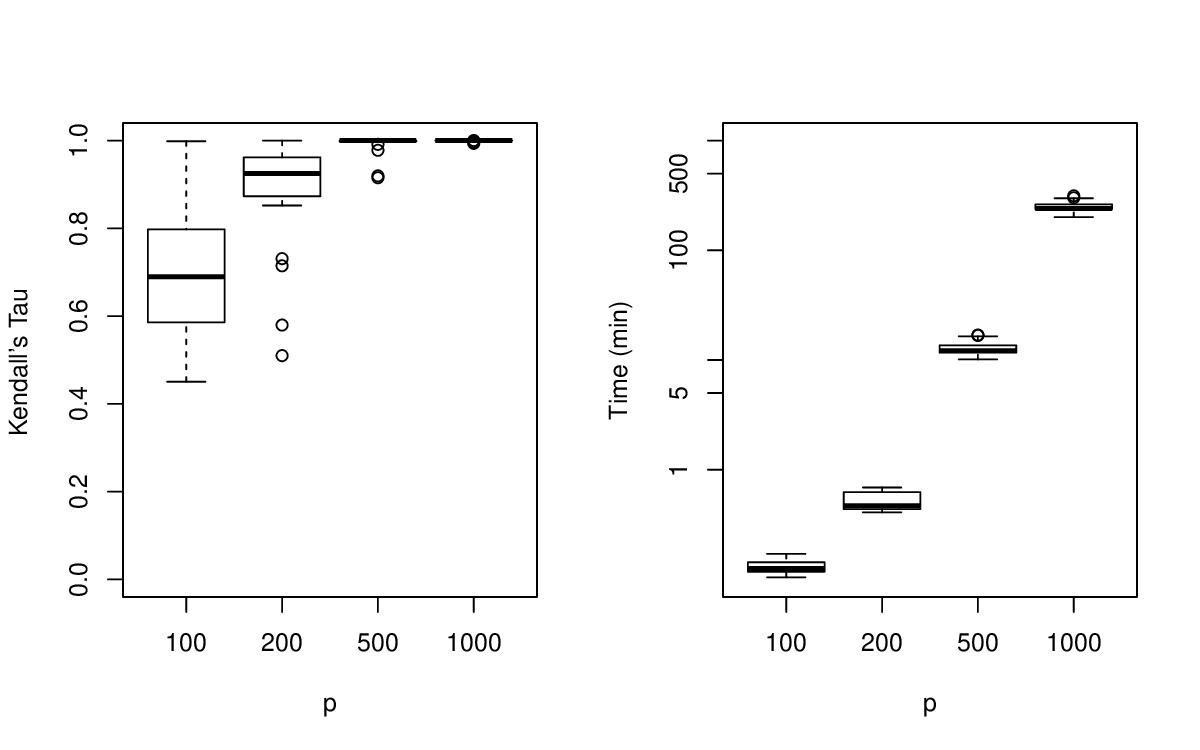}
	\caption{Results of graph estimation using the two step procedure where parents are selected via best subset regression.}
\end{figure}

\end{document}